\documentclass[11pt]{article}

\usepackage{fullpage}
\usepackage{amsmath, amsfonts, amssymb, bbm, mathtools, verbatim}
\usepackage{amsthm}
\usepackage{hyperref}
\usepackage{enumerate}
\usepackage[framemethod=tikz]{mdframed}
\usepackage{graphicx}
\usepackage{tikz-cd}
\usepackage{soul, xcolor}
\usepackage[ruled, linesnumbered]{algorithm2e}
\usepackage{thmtools}
\usepackage{thm-restate}
\usepackage{indentfirst} %Unindented first paragraphs are weird

\hypersetup{
colorlinks = true,
citecolor= blue
}



\newtheorem{theorem}{Theorem}
\newtheorem{lemma}[theorem]{Lemma}
\newtheorem{fact}[theorem]{Fact}
\newtheorem{claim}[theorem]{Claim}
\newtheorem{cor}[theorem]{Corollary}

\theoremstyle{definition}
\newtheorem{definition}{Definition}

\newcommand\supp{\mathop{\rm supp}}

\newcommand{\clA}{\mathcal{A}}
\newcommand{\clB}{\mathcal{B}}
\newcommand{\clC}{\mathcal{C}}
\newcommand{\clE}{\mathcal{E}}
\newcommand{\clF}{\mathcal{F}}

\newcommand{\clO}{\mathcal{O}}
\newcommand{\clP}{\mathcal{P}}

\newcommand{\clR}{\mathcal{R}}
\newcommand{\clS}{\mathcal{S}}
\newcommand{\clT}{\mathcal{T}}
\newcommand{\clU}{\mathcal{U}}
\newcommand{\clV}{\mathcal{V}}
\newcommand{\clX}{\mathcal{X}}
\newcommand{\clY}{\mathcal{Y}}
\newcommand{\clZ}{\mathcal{Z}}
\DeclareMathOperator*{\bbE}{\mathbb{E}}

\newcommand{\Id}{\mathbbm{1}}
\newcommand{\sfF}{\mathsf{F}}
\newcommand{\sfI}{\mathsf{I}}
\newcommand{\sfP}{\mathsf{P}}
\newcommand{\sfS}{\mathsf{S}}
\newcommand{\sfV}{\mathsf{V}}

\newcommand{\eps}{\varepsilon}

\newcommand{\Qi}{\mathrm{Q}^1}

\newcommand{\vph}{\varphi}

\newcommand{\ket}[1]{|#1\rangle}

\newcommand{\outprod}[2]{|#1 \rangle\langle #2|}
\newcommand{\inprod}[2]{\langle #1 | #2\rangle}
\newcommand{\matel}[3]{\langle #1 | #2 | #3\rangle}
\newcommand{\state}[1]{\outprod{#1}{#1}}
\newcommand{\Tr}{\mathrm{Tr}}

\allowdisplaybreaks

\title{A Direct Product Theorem for One-Way Quantum Communication}
\author{Rahul Jain \thanks{Centre for Quantum Technologies and Department of Computer Science, National University of Singapore and MajuLab, UMI 3654, Singapore. Email:~\tt{rahul@comp.nus.edu.sg}} \and
Srijita Kundu \thanks{Centre for Quantum Technologies, National University of Singapore, Singapore. Email:~\tt{srijita.kundu@u.nus.edu}}}
\date{}

\begin{document}	
\maketitle
\begin{abstract}
We prove a direct product theorem for the one-way entanglement-assisted quantum communication complexity of a general relation $f\subseteq\clX\times\clY\times\clZ$. For any $\eps, \zeta > 0$ and any $k\geq1$, we show that
\[ \Qi_{1-(1-\eps)^{\Omega(\zeta^6k/\log|\clZ|)}}(f^k) = \Omega\left(k\left(\zeta^5\cdot\Qi_{\eps + 12\zeta}(f) - \log\log(1/\zeta)\right)\right),\]
where $\Qi_{\eps}(f)$ represents the one-way entanglement-assisted quantum communication complexity of $f$ with worst-case error $\eps$ and $f^k$ denotes $k$ parallel instances of $f$.

As far as we are aware, this is the first direct product theorem for quantum communication -- direct sum theorems were previously known for one-way quantum protocols. Our techniques are inspired by the parallel repetition theorems for the entangled value of two-player non-local games, under product distributions due to Jain, Pereszl\'{e}nyi and Yao \cite{JPY14}, and under anchored distributions due to Bavarian, Vidick and Yuen \cite{BavarianVY17}, as well as message-compression for quantum protocols due to Jain, Radhakrishnan and Sen \cite{JRS03}. In particular, we show that a direct product theorem holds for the distributional one-way quantum communication complexity of $f$ under any distribution $q$ on $\clX\times\clY$ that is anchored on one side, i.e., there exists a $y^*$ such that $q(y^*)$ is constant and $q(x|y^*) = q(x)$ for all $x$. This allows us to show a direct product theorem for general distributions, since for any relation $f$ and any distribution $p$ on its inputs, we can define a modified relation $\tilde{f}$ which has an anchored distribution $q$ close to $p$, such that a protocol that fails with probability at most $\eps$ for $\tilde{f}$ under $q$ can be used to give  a protocol that fails with probability at most $\eps + \zeta$ for $f$ under $p$.

Our techniques also work for entangled non-local games which have input distributions anchored on any one side, i.e., either there exists a $y^*$ as previously specified, or there exists an $x^*$ such that $q(x^*)$ is constant and $q(y|x^*) = q(y)$ for all $y$. In particular, we show that for any game $G = (q, \clX\times\clY, \clA\times\clB, \sfV)$ where $q$ is a distribution on $\clX\times\clY$ anchored on any one side with anchoring probability $\zeta$, then
\[ \omega^*(G^k) = \left(1 - (1-\omega^*(G))^5\right)^{\Omega\left(\frac{\zeta^2 k}{\log(|\clA|\cdot|\clB|)}\right)}\]
where $\omega^*(G)$ represents the entangled value of the game $G$. This is a generalization of the result of \cite{BavarianVY17}, who proved a parallel repetition theorem for games anchored on both sides, i.e., where both a special $x^*$ and a special $y^*$ exist, and potentially a simplification of their proof.
\end{abstract}

\section{Introduction}
A fundamental question in complexity theory is: given $k$ independent instances of a function or relation, does computing them require $k$ times the amount of resources required to compute a single instance of the function or relation? Suppose solving one instance of some problem with success probability at least $p$ requires $c$ units of some resource. A natural way to solve $k$ independent instances of this problem would be to solve them independently, which requires $ck$ units of the resource. A \emph{direct sum theorem} for this problem would state that any algorithm for solving $k$ instances which uses $o(ck)$ units of resource has success probability at most $O(p)$. A \emph{direct product theorem} for the problem would state that any algorithm for solving $k$ instances that uses $o(ck)$ units of resource has success probability at most $p^{\Omega(k)}$. Hence a direct product theorem is the stronger result of the two.

In this paper, we deal with direct product theorems in the model of communication complexity. In this model, there are two parties Alice and Bob, who receive inputs $x$ and $y$ respectively, and wish to jointly compute a relation $f$. They can use local computation, public coins, and communicate with each other using classical messages, in the classical model; use local unitaries, shared entanglement, and communicate with each other using quantum messages, in the quantum model. The resource of interest is the number of bits/qubits communicated; so the parties are allowed to share an arbitrary amount of randomness or entanglement, and perform local operations of arbitrary complexity.

Direct product theorems in communication are related to \emph{parallel repetition theorems} for \emph{non-local games}. In a non-local game, two parties Alice and Bob are given inputs $x$ and $y$ respectively from some specified distribution, and without communicating with each other, they are required to give answers $a$ and $b$ respectively to a referee. They are considered to win the game if $\sfV(a,b,x,y)$ holds for a specified predicate $\sfV$. In the classical model, the players are allowed to share randomness, and in the quantum model they are allowed to share entanglement. A parallel repetition theorem shows that the maximum probability of winning $k$ independent instances of a non-local game is $p^{\Omega(k)}$, if the maximum probability of winning a single instance of it is $p$, regardless of the amount of shared randomness or entanglement used. Direct product theorems in communication are often proved by combining techniques used to prove direct sum theorems in communication, which require message-compression, and parallel repetition theorems for games.

In classical communication complexity, there is a long line of works on direct sum and direct-product theorems including~\cite{Razborov1992,ChakrabartiSWY01,BJKS02,Shaltiel2003,JainRS03,JainRS03a,JainRS08a, Klauck2007,Ben-Aroya2008,Lee2008,Viola2008,JainKN08,JainK09,HarshaJMR10,Klauck2010,JainY12,Sherstov2012,BarakBCR13,BravermanRWY13,Braverman2013a,BravermanR14,Braverman15,BravermanW15,Jain15,JainPY16,Gillat16,BravermanG18,Sherstov18}. A parallel repetition theorem for the classical value of general two-player non-local games was first shown by Raz \cite{Raz95}, and the proof was subsequently simplified by Holenstein \cite{Hol09}. 

In quantum communication complexity, a direct sum theorem is known for the entanglement-assisted one-way \cite{JainRS08a}, {\em simultaneous-message-passing} (SMP), entanglement-assisted 
\cite{JainRS08a} and unassisted models \cite{JainK09}. A strong parallel repetition theorem for the quantum value of a general two-player non-local game is not known. Parallel repetition theorems were shown for special classes of games such as XOR games \cite{CSUU08}, unique games \cite{KRT10} and projection games \cite{DSV15}. When the type of game is not restricted but the input distribution is, parallel repetition theorems have been shown under product distributions \cite{JPY14} and {\em anchored} distributions \cite{BavarianVY17, BVY15}. For general games under general distributions, the best current result is due to Yuen \cite{Yuen16}, which shows that the quantum value of $k$ parallel instances of a general game goes down polynomially in $k$, if the quantum value of the original game is strictly less than 1. No direct product theorems for quantum communication have so far been shown.

Using ideas from  Jain, Pereszl\'{e}nyi  and  Yao~\cite{JPY14} and the message-compression scheme from  Jain, Radhakrishnan and Sen~\cite{JainRS08a}, a strong direct product theorem for one-way quantum communication under product distributions can be shown. To deal with non-product distributions, we borrow the idea of anchored distributions due to Bavarian, Vidick and Yuen~\cite{BavarianVY17, BVY15}.

\subsection{Our results}
Let $\Qi_\eps(f)$ denote that the one-way entanglement-assisted quantum communication complexity of a relation $f$, with worst-case error $\eps$. Let $f^k$ denote $k$ parallel instances of $f$. Our strong direct product theorem is as follows.
\begin{restatable}{theorem}{main}
\label{thm:dpt}
For any relation $f\subseteq \clX\times\clY\times\clZ$, and any $\eps, \zeta > 0$,
\[ \Qi_{1-(1-\eps)^{\Omega(\zeta^6k/\log|\clZ|)}}(f^k) = \Omega\left(k\left(\zeta^5\cdot\Qi_{\eps + 12\zeta}(f) - \log\log(1/\zeta)\right)\right).\] 
\end{restatable}

Let $\omega^*(G)$ represent the entangled value of a two-player non-local game $G$, and let $G^k$ denote $k$ parallel instances of $G$. We call a distribution $q$ on $\clX\times\clY$ \emph{anchored on one side} with \emph{anchoring probability} $\zeta$ if one of the following conditions holds:
\begin{enumerate}[(i)]
\item There exists an $x^* \in \clX$ such that $q(x^*)=\zeta$ and $q(y|x^*) = q(y)$ for all $y\in\clY$,
\item There exists an $y^* \in \clY$ such that $q(y^*)=\zeta$ and $q(x|y^*) = q(x)$ for all $x\in\clX$. 
\end{enumerate}
The game will be called \emph{anchored on both sides} with anchoring probability $\zeta$ if both conditions hold instead.

Then our parallel repetition theorem is stated as follows.
\begin{restatable}{theorem}{parrep}
\label{thm:par-rep}
For a two-player non-local game $G = (q, \clX\times\clY, \clA\times\clB, \sfV)$ such that $q$ is a distribution anchored on one side with anchoring probability $\zeta$,
\[ \omega^*(G^k) = \left(1 - (1-\omega^*(G))^5\right)^{\Omega\left(\frac{\zeta^2 k}{\log(|\clA|\cdot|\clB|)}\right)}.\]
\end{restatable}
One can get a game anchored on one side (say the $\clY$ side) from a general game in the following way: in the anchored game, the referee chooses $(x,y)$ from the original probability distribution, and with probability $\zeta$ replaces $y$ with a new input $y^*$. If Bob's input is $y^*$, then the referee accepts any answer from the players. In a game anchored on both sides, the referee must instead replace $x$ with $x^*$ and $y$ with $y^*$ independently with probability $\zeta$, and accept if either Alice's input is $x^*$ or Bob's input is $y^*$. It is clear that anchoring makes the game easier. In this light, a parallel repetition theorem for anchoring games can be thought of as follows: for a general game $G$, there exists a simple transformation taking it to another game $\tilde{G}$ such that
\begin{enumerate}
\item If $\omega^*(G) = 1$, then $\omega^*(\tilde{G}^k) = 1$.
\item If $\omega^*(G) < 1$, then $\omega^*(\tilde{G}^k) = \exp(-\Omega(k))$.
\end{enumerate}
The merit of our result here is that the transformation involved for anchoring on one side changes the game less than the transformation involved in anchoring it on both sides.

We note that the definition of anchoring used on \cite{BavarianVY17, BVY15} is more general: instead of single inputs $x^*, y^*$, they consider anchoring sets $\clX^* \subseteq \clX$ and $\clY^* \subseteq \clY$, such that $q(\clX^*), q(\clY^*) \geq \zeta$, and whenever $x\in\clX^*$ or $y \in \clY^*$, $q(x,y) = q(x)q(y)$. However, it appears this generalized definition is not more useful from the perspective of anchoring transformations. While our technique could go through for the one-sided version of this definition of anchoring, we do not state or prove it as such for the sake of simplicity.

Unlike in the case of communication, worst-case success probability is usually not considered for non-local games. But one could define a game $G_{\text{wc}} = (\clX\times\clY, \clA\times\clB, \sfV)$ without an associated distribution, and the worst-case winning probability $\omega^*_\text{wc}$ of this over all inputs of this can be considered. As long as Alice and Bob are allowed to share randomness (which they are, in the quantum case), Yao's lemma \cite{Yao79} holds just like in the case of communication, relating the worst-case winning probability to distributional winning probability. Hence, by choosing $\zeta = (1 - \omega^*_\text{wc}(G_\text{wc}))/2$ and using the same arguments as in the case of communication, Theorem \ref{thm:par-rep} leads to the following corollary about the worst-case winning probability of any game.
\begin{cor}
For any two-player non-local game $G_\mathrm{wc} = (\clX\times\clY, \clA\times\clB, \sfV)$,
\[ \omega^*_\mathrm{wc}(G^k_\mathrm{wc}) = \left(1-(1-\omega^*_\mathrm{wc}(G_\mathrm{wc}))^7\right)^{\Omega\left(\frac{k}{\log(|\clA|\cdot|\clB|)}\right)}.\]
\end{cor}

\subsection{Proof overview}
We use the information theoretic framework for parallel repetition and direct product theorems established by \cite{Raz95} and \cite{Hol09}. The broad idea is as follows: for a given relation $f \subseteq \clX\times\clY\times\clZ$, let the one-way quantum communication required to compute a single copy with constant success be $c$. Now consider a one-way quantum protocol $\clP$ for $f^k$ which has communication $o(ck)$, in which we can condition on the success of some $t$ coordinates. If the success probability in these $t$ coordinates is already as small as we want, then we are done. Otherwise, we exhibit a $(t+1)$-th coordinate $i$, such that conditioned on the success on the $t$ coordinates, the success of $i$ in $\clP$ is bounded away from 1. This is done by showing that if the success probability in the $t$ coordinates in not too small, then we can give a protocol $\clP'$ for $f$ whose communication is $o(c)$ and whose success probability is constant -- a contradiction.

$\clP'$ works by embedding its input into the $i$-th coordinate of a shared quantum state representing the final input, output, message and discarded registers of $\clP$, conditioned on the success event in the $t$ coordinates, which we denote by $\clE$. Suppose the quantum state conditioned on $\clE$, when Alice and Bob's inputs are $x_i$ and $y_i$ respectively at the $i$-th coordinates, is $\ket{\vph}_{x_iy_i}$. On input $(x_i,y_i)$ in $\clP'$, Alice and Bob will by means of local unitaries and communication try to get the shared state close to $\ket{\vph}_{x_iy_i}$, on which Bob can perform a measurement to get an outcome $z_i$. The state $\ket{\vph}_{x_iy_i}$ is such that the resulting probability distribution $\sfP_{X_iY_iZ_i}$ is the distribution of $X_iY_iZ_i$ in $\clP$ conditioned on success. Hence our proof mainly consists of showing how Alice and Bob can get the shared state close to $\ket{\vph}_{x_iy_i}$.

The proof technique for a parallel repetition theorem is same, except one cannot, and need not, use communication to get the shared state $\ket{\vph}_{x_iy_i}$ there. In order to motivate our techniques, we shall briefly describe the techniques used in \cite{JPY14} and \cite{BVY15} to get $\ket{\vph}_{x_iy_i}$.
\begin{itemize}
\item In \cite{JPY14} the following three states are considered: $\ket{\vph}_{x_i}$ which is the superposition of $\ket{\vph}_{x_iy_i}$ over the distribution of $Y_i$, $\ket{\vph}_{y_i}$ which is the superposition over the distribution of $X_i$,  and $\ket{\vph}$ which is the superposition over both. In this setting, $X_1\ldots X_k$ are initially in product with all of Bob's registers and $Y_1\ldots Y_k$ are in product with all of Alice's registers. If the probability of $\clE$ is large, then conditioning on it, the following can be shown:
\begin{enumerate}
\item By chain rule of mutual information, there is an $X_i$ whose mutual information with Bob's registers in $\ket{\vph}$ is small. Hence by Uhlmann's theorem, there exist unitaries $U_{x_i}$ acting on Alice's registers that take $\ket{\vph}$ close to $\ket{\vph}_{x_i}$.
\item Similarly, the mutual information between $Y_i$ and Alice's registers in $\ket{\vph}$ is small, and hence there exist unitaries $U_{y_i}$ acting on Bob's registers that take $\ket{\vph}$ close to $\ket{\vph}_{y_i}$.
\item Since $U_{x_i}$ and $U_{y_i}$ act on disjoint registers, using a {\em commuting} argument and the {\em monotonicity} of trace-distance under quantum-operations, $U_{x_i}\otimes U_{y_i}$ takes $\ket{\vph}$ close to $\ket{\vph}_{x_iy_i}$.
\end{enumerate}
Alice and Bob can thus share $\ket{\vph}$ as entanglement, and get close to $\ket{\vph}_{x_iy_i}$ by local operations.

\item In \cite{BVY15}, $X_1\ldots X_k$ are not initially in product with $Y_1\ldots Y_k$, hence they need to use what are known as \emph{correlation-breaking variables}. For each $i$, correlation-breaking variables $D_iG_i$ are such that conditioned on $D_iG_i$, $X_i$ and $Y_i$ are independent. In particular, $D_i$ is a uniformly distributed bit, and $G_i$ takes values in either $\clX$ or $\clY$ depending on whether $D_i$ is 0 or 1, and is highly correlated with either $X_i$ or $Y_i$ in the respective cases. This means that conditioned on $D_i=0$, $G_i=x^*$ with probability $\Omega(\zeta)$ and conditioned on $D_i=1$, $G_i=y^*$ with probability $\Omega(\zeta)$.
\begin{enumerate}
\item The mutual information between $X_i$ and Bob's registers in $\ket{\vph}$ conditioned on $D_i=1$ and $G_i$ is small. Further conditioning on $G_i=y^*$ (which happens with constant probability), the mutual information between $X_i$ and Bob's registers in $\ket{\vph}_{y^*}$ is small. Hence by Uhlmann's theorem, there exist unitaries $U_{x_i}$ on Alice's registers, taking $\ket{\vph}_{x^*y^*}$ close to $\ket{\vph}_{x_iy^*}$.
\item Similarly, the mutual information between $Y_i$ and Alice's registers in $\ket{\vph}$ conditioning on $D_i=0$ and $G_i=x^*$ is small, which means there exist unitaries $U_{y_i}$ on Bob's registers, taking $\ket{\vph}_{x^*y^*}$ close to $\ket{\vph}_{x^*y_i}$.
\item Using an involved argument, it is possible to show that $U_{x_i}\otimes U_{y_i}$ takes $\ket{\vph}_{x^*y^*}$ close to $\ket{\vph}_{x_iy_i}$.
\end{enumerate}
Alice and Bob can thus share $\ket{\vph}_{x^*y^*}$ in this case, and get close to $\ket{\vph}_{x_iy_i}$ by local operations.
\end{itemize}

In our direct product proof, since the distribution is anchored on one side, we use correlation-breaking variables that are identical to those in \cite{BVY15} in the $D_i=1$ case, but in the $D_i=0$ we consider a simpler distribution where $G_i$ is perfectly correlated with $X_i$. Here we also clarify what we mean by $G_i$ and $Y_i$ being highly correlated when $D_i=1$: if $G_i=y^*$, then $Y_i$ is always $y^*$; but if $G_i=y_i$ for $y_i\neq y^*$, then $Y_i$ still takes value $y^*$ with probability $\Omega(\zeta)$, and is $y_i$ otherwise. The distribution of $X_i$ conditioned on $G_i=y^*$ is the marginal distribution of $X_i$, while conditioned on $y_i$, it is the same as the distribution of $X_i$ conditioned on $Y_i=y_i$ (potentially different from the marginal distribution of $X_i$). Our use of these correlation-breaking variables is quite different from that in \cite{BVY15}, however.

We note that in a communication protocol where Alice sends the message, we cannot hope to show that the mutual information between $X_i$ and Bob's registers is small even conditioned on the the correlation-breaking variables, since the final state on Bob's side includes the message from Alice, which can potentially be fully correlated with Alice's inputs. Since Bob does not communicate however, the same does not apply to him. Hence we can show the following:
\begin{enumerate}
\item If the message size is $o(ck)$, by chain rule of mutual information, the mutual information between $X_i$ and Bob's registers in $\ket{\vph}$ is $o(c)$, conditioned on $D_i=1,G_i=y^*$. Since the distribution is anchored on Bob's side, this means that the mutual information between $X_i$ and Bob's registers in $\ket{\vph}_{y^*}$ is $o(c)$. Using a result from \cite{JainRS02,JainRS08a}, then there exist projectors $\Pi_{x_i}$ acting on Alice's registers, which succeed with probability $2^{-o(c)}$ on $\ket{\vph}_{y^*}$, and on success take it close to $\ket{\vph}_{x_iy^*}$.
\item The mutual information between $Y_i$ and Alice's registers conditioned on $D_i=1,G_i\neq y^*$ is small. For each value of $G_i\neq y^*$, there exist only two possible values of $Y_i$: $y_i$ and $y^*$, and hence Alice's registers in $\ket{\vph}_{y_i}$ and $\ket{\vph}_{y^*}$ must be close on average. By Uhlmann's theorem, there exist unitaries $U_{y_i}$ acting on Bob's registers, taking $\ket{\vph}_{y^*}$ close to $\ket{\vph}_{y_i}$.
\item Since the marginal distribution of $X_i$ conditioned on $G_i=y_i$ is approximately the same as the marginal distribution of $X_i$ conditioned on $Y_i=y_i$, we can show by the same argument as in \cite{JainRS08a,JPY14}, that conditioned on success of $\Pi_{x_i}$, $\Pi_{x_i}\otimes U_{y_i}$ takes $\ket{\vph}_{y^*}$ close to $\ket{\vph}_{x_iy_i}$.
\end{enumerate} 
Hence there is a communication protocol with prior shared entanglement which allows Alice and Bob to obtain a state close to $\ket{\vph}_{x_iy_i}$ as a shared state on input $(x_i,y_i)$: Alice and Bob share $2^{o(c)}$ copies of $\ket{\vph}_{y^*}$ as entanglement; Alice performs the $\Pi_{x_i}$ measurement on all these copies, and succeeds on at least one copy with high probability. She sends the index of the copy on which she succeeds to Bob, who performs $U_{y_i}$ on the same copy. This protocol has communication $o(c)$, since that is how many classical bits Alice needs in order to encode the index of the successful copy out of $2^{o(c)}$ copies. This completes the proof of the direct product theorem.

Our parallel repetition proof is same as above, except no communication is necessary, since there was no communication in the original protocol. Instead of a projector on Alice's registers taking $\ket{\vph}_{y^*}$ close to $\ket{\vph}_{x_iy^*}$, in this case we will have a unitary $U_{x_i}$ doing it. We can argue identically to the direct product proof that there exist $U_{y_i}$ taking $\ket{\vph}_{y^*}$ close to $\ket{\vph}_{y_i}$, and $U_{x_i}\otimes U_{y_i}$ takes $\ket{\vph}_{y^*}$ close to $\ket{\vph}_{x_iy_i}$. The last part, indicated as step 3 above, is arguably simpler in our proof compared to \cite{BVY15}.

%-------------------------------------------

\section{Preliminaries}
\subsection{Probability theory}
We shall denote the probability distribution of a random variable $X$ on some set $\clX$ by $\sfP_X$. For any event $\clE$ on $\clX$, the distribution of $X$ conditioned on $\clE$ will be denoted by $\sfP_{X|\clE}$. For joint random variables $XY$, $\sfP_{X|Y=y}(x)$ is the conditional distribution of $X$ given $Y=y$; when it is clear from context which variable's value is being conditioned on, we shall often shorten this to $\sfP_{X|y}$. We shall use $\sfP_{XY}\sfP_{Z|X}$ to refer to the distribution
\[ (\sfP_{XY}\sfP_{Z|X})(x,y,z) = \sfP_{XY}(x,y)\cdot\sfP_{Z|X=x}(z).\]
For two distributions $\sfP_X$ and $\sfP_{X'}$ on the same set $\clX$, the $\ell_1$ distance between them is defined as
\[ \Vert\sfP_X - \sfP_{X'}\Vert_1 = \sum_{x\in\clX}|\sfP_X(x) - \sfP_{X'}(x)|.\]

\begin{fact}
For joint distributions $\sfP_{XY}$ and $\sfP_{X'Y'}$ on the same sets,
\[ \Vert\sfP_X -  \sfP_{X'}\Vert_1 \leq \Vert\sfP_{XY} - \sfP_{X'Y'}\Vert_1.\]
\end{fact}
\begin{fact}\label{l1-dist}
For two distributions $\sfP_X$ and $\sfP_{X'}$ on the same set and an event $\clE$ on the set,
\[ |\sfP_X(\clE) - \sfP_{X'}(\clE)| \leq \frac{1}{2}\Vert\sfP_X - \sfP_{X'}\Vert_1.\]
\end{fact}
\begin{fact}\label{coupling}
For two distributions $\sfP_X$ and $\sfP_{X'}$ on the same set, and any joint distribution $\sfP_{XX'}$ whose marginals are $\sfP_X$ and $\sfP_{X'}$ respectively, we have
\[ \Vert\sfP_X - \sfP_{X'}\Vert_1 \leq 2\sfP_{XX'}(X\neq X').\]
\end{fact}
\begin{fact}\label{cond-prob}
Suppose probability distributions $\sfP_X, \sfP_{X'}$ satisfy $\Vert \sfP_X - \sfP_{X'}\Vert_1 \leq \eps$, and an event $\clE$ satisfies $\sfP_X(\clE) \geq \alpha$, where $\alpha > \eps$. Then,
\[ \Vert\sfP_{X|\clE} - \sfP_{X'|\clE}\Vert_1 \leq \frac{2\eps}{\alpha}.\]
\end{fact}
\begin{proof}
From Fact \ref{l1-dist}, $\alpha-\eps/2 \leq \sfP_{X'}(\clE) \leq \alpha+\eps/2$. By definition, there exists an event $\clE'$ such that $2(\sfP_{X|\clE}(\clE') - \sfP_{X'|\clE}(\clE')) = \Vert\sfP_{X|\clE} - \sfP_{X'|\clE}\Vert_1$. Now, $\sfP_{X}(\clE\land\clE') = \sfP_X(\clE)\sfP_{X|\clE}(\clE') \geq \alpha\sfP_{X|\clE}(\clE')$. Similarly, $\sfP_{X'}(\clE\land\clE') \leq (\alpha+\eps/2)\sfP_{X'|\clE}(\clE') \leq \alpha\sfP_{X'|\clE}(\clE') + \frac{1}{2}\Vert\sfP_X - \sfP_{X'}\Vert_1$.

Now,
\begin{align*}
\Vert\sfP_{X} - \sfP_{X'}\Vert_1 & \geq 2(\sfP_X(\clE\land \clE') - \sfP_{X'}(\clE\land\clE')) \\
 & \geq 2\alpha(\sfP_{X|\clE}(\clE') - \sfP_{X'|\clE}(\clE')) - \Vert\sfP_X - \sfP_{X'}\Vert_1 \\
 & \geq \alpha\Vert\sfP_{X|\clE} - \sfP_{X'|\clE}\Vert_1 - \Vert\sfP_X - \sfP_{X'}\Vert_1
\end{align*}
which gives the required result.
\end{proof}

\begin{fact}[\cite{BVY15}, Lemma 16]\label{anchor}
Suppose $XYZ$ are random variables satisfying $\sfP_{XY}(x,y^*)=\alpha\cdot\sfP_X(x)$ for all $x$. Then,
\[ \left\Vert \sfP_{XYZ} - \sfP_{XY}\sfP_{Z|X,y^*}\right\Vert_1 \leq \frac{2}{\alpha}\left\Vert \sfP_{XYZ}-\sfP_{XY}\sfP_{Z|X}\right\Vert_1.\]
\end{fact}

\begin{cor}\label{x*y*}
Supose $\sfP_{XY}$ and $\sfP_{X'Y'Z'}$ are distributions such that $\Vert \sfP_{XY} - \sfP_{X'Y'}\Vert_1 \leq \eps$, and $\sfP(x,y^*) = \alpha\cdot\sfP_X(x)$ for all $x$. Then,
\[ \Vert\sfP_{X'Z'|y^*} - \sfP_{X'Z'}\Vert_1 \leq \frac{11}{\alpha}\Vert\sfP_{X'Y'Z'} - \sfP_{XY}\sfP_{Z'|X'}\Vert_1.\]
\end{cor}
\begin{proof}
Let $\Vert \sfP_{X'Y'Z'} - \sfP_{XY}\sfP_{Z'|X'}\Vert_1 = \eps$. Note that
\[
\Vert \sfP_{X|y^*} - \sfP_{X'|y^*}\Vert_1 \leq \frac{2\eps}{\alpha}
\]
by Fact \ref{cond-prob}. Let $\sfP_{XYZ''}$ denote the distribution $\sfP_{XY}\sfP_{Z'|X'Y'}$.
\begin{align*}
\Vert \sfP_{X'Z'} - \sfP_{XZ''}\Vert_1 & = \sum_{x,z}\left| \sfP_{X'}(x)\sum_y\sfP_{Y'|x}(y)\sfP_{Z'|xy}(z) - \sfP_X(x)\sum_y\sfP_{Y|x}(y)\sfP_{Z'|xy}(z)\right| \\
 & \leq \sum_{x,y,z}\left|\sfP_{X'}(x)\sfP_{Y'|x}(y) - \sfP_X(x)\sfP_{Y|x}(y)\right|\sfP_{Z'|xy}(z) \\
 & = \Vert \sfP_{X'Y'} - \sfP_{XY}\Vert_1 \leq \eps.
\end{align*}
\begin{align*}
\Vert \sfP_{XYZ''} - \sfP_{XY}\sfP_{Z''|X}\Vert_1 & \leq \Vert\sfP_{XYZ''} - \sfP_{X'Y'Z'}\Vert_1 + \Vert\sfP_{X'Y'Z'} - \sfP_{XY}\sfP_{Z'|X'}\Vert_1 \\
& \quad + \Vert\sfP_{XY}\sfP_{Z'|X'} - \sfP_{XY}\sfP_{Z''|X}\Vert_1 \\
 & = \Vert \sfP_{XY} - \sfP_{X'Y'}\Vert_1 + \Vert\sfP_{X'Y'Z'} - \sfP_{XY}\sfP_{Z'|X'}\Vert_1 \\
& \quad + \sum_{x,y}\sfP_{XY}(x,y)\Vert \sfP_{Z'|x} - \sfP_{Z''|x}\Vert_1 \\
 & \leq 2\eps + \sum_x\sfP_X(x)\sum_{y,z}|\sfP_{Y|x}(y) - \sfP_{Y'|x}(y)|\sfP_{Z'|xy}(z) \\
 & \leq 2\eps + \sum_{x,y}|\sfP_X(x)\sfP_{Y|x}(y) - \sfP_{X'}(x)\sfP_{Y'|x}(y)| \\
& \quad + \sum_{x,y}|\sfP_{X'}(x) - \sfP_X(x)|\sfP_{Y'|x}(y) \\
 & \leq 2\eps + 2\Vert\sfP_{XY} - \sfP_{X'Y'}\Vert_1 \leq 4\eps.
\end{align*}
Combining all this,
\begin{align*}
\Vert \sfP_{X'Z'|y^*} - \sfP_{X'Z'}\Vert_1 & \leq \Vert \sfP_{X'Z'|y^*} - \sfP_{XZ''|y^*}\Vert_1 + \Vert \sfP_{XZ''|y^*} - \sfP_{XZ''}\Vert_1 + \Vert \sfP_{XZ''} - \sfP_{X'Z'}\Vert_1 \\
 & \leq \Vert \sfP_{X|y^*} - \sfP_{X'|y^*}\Vert_1 + \Vert \sfP_{XZ''|y^*} - \sfP_{XZ''}\Vert_1 + \Vert \sfP_{XZ''} - \sfP_{X'Z'}\Vert_1 \\
 & \leq \frac{2\eps}{\alpha} + \frac{2}{\alpha}\Vert \sfP_{XYZ''} - \sfP_{XY}\sfP_{Z''|X}\Vert_1 + \eps\\
 & \leq  \frac{2\eps}{\alpha} + \frac{8\eps}{\alpha} + \eps \leq \frac{11\eps}{\alpha}.
\end{align*}
where we have used Lemma \ref{anchor} in the third inequality.
\end{proof}

\begin{fact}[\cite{Hol09}, Corollary 6]\label{hol}
Let $\sfP_{TU_1\ldots U_kV} = \sfP_T\sfP_{U_1|T}\sfP_{U_2|T}\ldots\sfP_{U_k|T}\sfP_{V|TU_1\ldots U_k}$ be a probability distribution over $\clT\times\clU^k\times\clV$, and let $\clE$ be any event. Then,
\[ \sum_{i=1}^k\Vert\sfP_{TU_iV|\clE}-\sfP_{TV|\clE}\sfP_{U_i|T}\Vert_1 \leq \sqrt{k\left(\log(|\clV|) + \log\left(\frac{1}{\Pr[\clE]}\right)\right)}.\] 
\end{fact}

\begin{definition}[\cite{Hol09}]
For two distributions $\sfP_{XY}$ and $\sfP_{X'Y'ST}$, we say $(X,Y)$ is $(1-\eps)$-embeddable in $(X'S,Y'T)$ if there exists a random variable $R$ on a set $\clR$ independent of $XY$ and functions $f_A: \clX \times \clR \to \clS$ and $f_B: \clY\times\clR\to\clT$, such that
\[ \Vert\sfP_{XYf_A(X,R)f_B(X,R)} - \sfP_{X'Y'ST}\Vert_1 \leq \eps.\]
\end{definition}
\begin{fact}[\cite{Hol09, JainPY16}]\label{embed}
If two distributions $\sfP_{XY}$ and $\sfP_{X'Y'R'}$ satisfy
\[ \Vert\sfP_{X'Y'R'} - \sfP_{XY}\sfP_{R'|X'}\Vert_1 \leq \eps \quad \quad \Vert\sfP_{X'Y'R'} - \sfP_{XY}\sfP_{R'|Y'}\Vert_1 \leq \eps,\]
then $(X,Y)$ is $(1-5\eps)$-embeddable in $(X'R',Y'R')$.\footnote{This fact is equivalent to Lemma 2.11 in \cite{JainPY16}, although this lemma is stated in terms of relative entropies instead of trace distances between the various distributions. In the proof of the lemma, the relative entropies are converted to the same trace distances as we consider, using Pinsker's inequality. This justifies our statement of the fact, which is tailored towards our application.}
\end{fact}

\subsection{Quantum information}
The $\ell_1$ distance between two quantum states $\rho$ and $\sigma$ is given by
\[ \Vert\rho-\sigma\Vert_1 = \Tr\sqrt{(\rho-\sigma)^\dagger(\rho-\sigma)} = \Tr|\rho-\sigma|.\]
The fidelity between two quantum states is given by
\[ \sfF(\rho,\sigma) = \Vert\sqrt{\rho}\sqrt{\sigma}\Vert_1.\]
$\ell_1$ distance and fidelity are related in the following way.
\begin{fact}[Fuchs-van de Graaf inequality]\label{fvdg}
For any pair of quantum states $\rho$ and $\sigma$,
\[ 2(1-\sfF(\rho,\sigma)) \leq \Vert\rho-\sigma\Vert_1\leq 2\sqrt{1-\sfF(\rho,\sigma)^2}.\]
For two pure states $\ket{\psi}$ and $\ket{\phi}$, we have
\[ \Vert\state{\psi} - \state{\phi}\Vert_1 = \sqrt{1 - \sfF\left(\state{\psi},\state{\phi}\right)^2} = \sqrt{1-|\inprod{\psi}{\psi}|^2}.\]
\end{fact}
\begin{fact}[Uhlmann's theorem]\label{uhlmann}
Suppose $\rho$ and $\sigma$ are mixed states on register $X$ which are purified to $\ket{\rho}$ and $\ket{\sigma}$ on registers $XY$, then it holds that
\[ \sfF(\rho, \sigma) = \max_U|\matel{\rho}{\Id_X\otimes U}{\sigma}|\]
where the maximization is over unitaries acting only on register $Y$.
\end{fact}
\begin{fact}\label{chan-l1}
For a quantum channel $\clE$ and states $\rho$ and $\sigma$,
\[ \Vert\clE(\rho) - \clE(\sigma)\Vert_1 \leq \Vert\rho-\sigma\Vert_1 \quad \quad \text{and} \quad \quad \sfF(\clE(\rho),\clE(\sigma)) \geq \sfF(\rho,\sigma).\]
\end{fact}
The entropy of a quantum state $\rho$ on a register $Z$ is given by
\[ \sfS(\rho) = -\Tr(\rho\log \rho).\]
The relative entropy between two states $\rho$ and $\sigma$ of the same dimensions is given by
\[ \sfS(\rho\Vert \sigma) = \Tr(\rho\log\rho) - \Tr(\rho\log\sigma).\]
The relative min-entropy between $\rho$ and $\sigma$ is defined as
\[ \sfS_\infty(\rho\Vert\sigma) = \min\{\lambda : \rho \leq 2^\lambda\sigma\}.\]
It is easy to see that $\sfS(\rho\Vert\sigma)$ and $\sfS_\infty(\rho\Vert\sigma)$ only take finite values when the support of $\rho$ is contained in the support of $\sigma$. Moreover, clearly $0 \leq \sfS(\rho\Vert\sigma) \leq \sfS_\infty(\rho\Vert\sigma)$ for all $\rho$ and $\sigma$.

The $\eps$-smooth relative min-entropy between $\rho$ and $\sigma$ is defined as
\[ \sfS^\eps_\infty(\rho\Vert \sigma) = \inf_{\rho': \Vert\rho-\rho'\Vert_1\leq\eps}\sfS(\rho'\Vert\sigma).\]
$\sfS^\eps_\infty(\rho\Vert\sigma)$ can take a finite value even if the support of $\rho$ is not contained in the support of $\sigma$, for example if $\rho$ is $\eps$-close to a state contained within the support of $\sigma$. $\sfS_\infty(\rho\Vert\sigma)$ cannot be upper bounded by $\sfS(\rho\Vert\sigma)$, but $\sfS^\eps_\infty(\rho\Vert\sigma)$ can be, due to the Quantum Substate Theorem.
\begin{fact}[Quantum Substate Theorem, \cite{JRS09,JN12}]
For any two states $\rho$ and $\sigma$ such that the support of $\rho$ is contained in the support of $\sigma$, and any $\eps > 0$,
\[ \sfS^\eps_\infty(\rho\Vert\sigma) \leq \frac{4\sfS(\rho\Vert\sigma)}{\eps^2} + \log\left(\frac{1}{1-\eps^2/4}\right).\]
\end{fact}
\begin{fact}[Pinsker's Inequality]\label{pinsker}
For any two states $\rho$ and $\sigma$, $\Vert \rho - \sigma\Vert_1 \leq \sqrt{\sfS(\rho\Vert\sigma)}$.
\end{fact}
\begin{fact}\label{event-prob}
If $\sigma = \eps\rho + (1-\eps)\rho'$, then $\sf\sfS_\infty(\rho\Vert \sigma) \leq \log(1/\eps)$.
\end{fact}
\begin{fact}\label{Sinfty-tri}
For any three quantum states $\rho, \sigma, \vph$ such that $\supp(\rho) \subseteq \supp(\vph) \subseteq \supp(\sigma)$,
\[ \sf\sfS_\infty(\rho\Vert\sigma) \leq \sf\sfS_\infty(\rho\Vert\vph) + \sfS_\infty(\vph\Vert\sigma).\]
\end{fact}
\begin{fact}\label{u-inv}
For any unitary $U$, $\sfS_\infty(U\rho U^\dagger\Vert U\sigma U^\dagger) = \sfS_\infty(\rho\Vert\sigma)$.
\end{fact}
A state of the form
\[ \rho_{XY} = \sum_x \sfP_X(x)\state{x}_X\otimes\rho_{Y|x}\]
is called a CQ (classical-quantum) state, with $X$ being the classical register and $Y$ being quantum. We shall use $X$ to refer to both the classical register and the classical random variable with the associated distribution. As in the classical case, here we are using $\rho_{Y|x}$ to denote the state of the register $Y$ conditioned on $X=x$, or in other words the state of the register $Y$ when a measurement is done on the $X$ register and the outcome is $x$. Hence $\rho_{XY|x} = \state{x}_X\otimes \rho_{Y|x}$. When the registers are clear from context we shall often write simply $\rho_x$.

The mutual information between $Y$ and $Z$ with respect to a state $\rho$ on $YZ$ is defined as
\[  \sfI(Y:Z)_\rho = \sfS(\rho_{YZ}\Vert\rho_Y\otimes\rho_Z).\]
The $\eps$-smooth max-information between $Y$ and $Z$ with respect to $\rho$ is defined as
\[ \sfI^\eps_{\max}(Y:Z)_\rho = \inf_{\sigma_Z}\sfS^\eps_\infty(\rho_{YZ}\Vert \rho_Y\otimes \sigma_Z).\]

The conditional mutual information between $Y$ and $Z$ conditioned on a classical register $X$, is defined as
\[  \sfI(Y:Z|X) = \bbE_{\sfP_X}[\sfI(Y:Z)_{\rho_x}].\]
Mutual information can be seen to satisfy the chain rule
\[  \sfI(XY:Z)_\rho =  \sfI(X:Z)_\rho +  \sfI(Y:Z|X)_\rho.\]

\begin{fact}[\cite{BCR11}, Lemma B.7]\label{dim-ub}
For any quantum state $\rho_{YZ}$,
\[ \inf_{\sigma_Z}\sfS_\infty(\rho_{YZ} \Vert \rho_Y\otimes\sigma_Z) \leq 2\min\{\log|\clY|,\log|\clZ|\}.\]
\end{fact}
\begin{fact}\label{cond-dec}
For CQ states
\[ \rho_{XY} = \sum_x\sfP_{X}(x)\state{x}_X\otimes\rho_{Y|x} \quad \quad \sigma_{XY} = \sum_x\sfP_{X'}(x)\state{x}_X\otimes\sigma_{Y|x},\]
their relative entropy is given by
\[ \sfS(\rho_{XY}\Vert \sigma_{XY}) = \sfS(\sfP_X\Vert \sfP_{X'}) + \bbE_{\sfP_X}[\sfS(\rho_{Y|x}\Vert \sigma_{Y|x})].\]
\end{fact}
\begin{fact}\label{Imax-close}
Suppose $\sigma_{XYZ}$ and $\rho_{XYZ}$ are CQ states defined as follows
\[ \sigma_{XYZ} = \sum_{x,y}\sfP_{XY}(x,y)\state{x,y}\otimes\sigma_{Z|xy} \quad \quad \rho_{XYZ} = \sum_{x,y}\sfP_{X'Y'}(x,y)\state{x,y}\otimes\sigma_{Z|xy}, \]
where $\Vert \sfP_{XY} - \sfP_{X'Y'}\Vert_1 \leq \delta <\frac{1}{2}$. Let $\sfI(Y:Z|X)_\sigma \leq c$. Then, for any $\delta < \eps < \frac{1}{2}$,
\[ \sfP_{X'}\left(\sfI^{\eps+7\delta/\eps}_{\max}(Y:Z)_{\rho_x} > \frac{4c+1}{\eps^3}\right) \leq 2\eps + \frac{\delta}{2}.\]
\end{fact}
\begin{proof}
Let $\text{Good}_1$ denote the set of $x$ such that $\sfI(Y:Z)_{\sigma_x} \leq c/\eps$. Due to Markov's inequality, $\sfI(Y:Z|X)_\rho \leq c$ implies $\sfP_X(\text{Good}_1) \geq 1- \eps$. By Quantum Substate Theorem, for each $x \in \text{Good}_1$, there exist a $\sigma'_{YZ|x}$ such that $\Vert \sigma_{YZ|x} - \sigma'_{YZ|x}\Vert_1 \leq \eps$, and a $\theta_{Z|x}$ such that
\[ \sfS_\infty(\sigma'_{YZ|x}\Vert\sigma_{Y|x}\otimes\theta_{Z|x}) \leq \frac{4c}{\eps^3} + \log\left(\frac{1}{1-\eps^2/4}\right) = k \text{ (say)}.\]
Clearly $\sigma_{Y|z} = \sum_y\sfP_{Y|X=x}(y)\state{y}_Y$ and $\rho_{Y|x} = \sum_y\sfP_{Y'|X'=x}(y)\state{y}_Y$. Let
\[ \clE_x = \left\{y: \log\frac{\sfP_{Y|x}(y)}{\sfP_{Y'|x}(y)} \leq 1\right\} \quad \quad \clE = \left\{(x,y): \log\frac{\sfP_{Y|x}(y)}{\sfP_{Y'|x}(y)} \leq 1 \right\}.\]
We observe that
\begin{align*}
\sfP_{XY}(\clE^c) & \leq \delta/2 + \sfP_{X'Y'}(\clE^c) \\
 & \leq \delta/2 + \frac{1}{2}\sum_{(x,y) \notin \clE}\sfP_{X'}(x)\sfP_{Y|x}(y) \\
 & \leq \delta/2 + \sfP_{XY}(\clE^c)/2 + \frac{1}{2}\sum_{(x,y) \notin \clE}|\sfP_{X'}(x) - \sfP_X(x)|\sfP_{Y|x}(y) \\
 & \leq \delta/2 + \sfP_{XY}(\clE^c)/2 + \frac{1}{2}\sum_x|\sfP_{X'}(x) - \sfP_X(x)| \\
 & \leq \delta + \sfP_{XY}(\clE^c)/2,
\end{align*}
which gives us $\sfP_{XY}(\clE^c) = \bbE_{\sfP_X}\sfP_Y(\clE^c_x) \leq 2\delta$. Let $\text{Good}_2$ denote the set of $x$ such that $\sfP(\clE^c_x) \leq 2\delta/\eps$. By Markov's inequality, $\sfP_X(\text{Good}_2) \geq 1 - \eps$.

Let $\Pi_x$ denote the projector on $Y$ that projects to the subset $\clE_x$. By definition, $\Pi_x\sigma_{Y|x}\Pi_x \leq 2\rho_{Y|x}$. Now $\sigma'_{YZ|x} \leq 2^k(\sigma_{Y|x}\otimes\theta_{Z|x})$ implies
\[ \Pi_x\sigma'_{Y|x}\Pi_x \leq 2^k\Pi_x(\sigma_{Y|x}\otimes\theta_{Z|x})\Pi_x \leq 2^{k+1}\rho_{Y|x}\otimes\theta_{Z|x}.\]
Let $\rho'_{YZ|x}$ denote $\Pi_x\sigma'_{Y|x}\Pi_x/\Tr(\Pi_x\sigma'_{YZ|x}\Pi_x)$. We note
\begin{align*}
\Tr(\Pi_x\sigma'_{YZ|x}\Pi_x) & \geq \Tr(\Pi_x\sigma_{YZ|x}\Pi_x) - \Tr(\Pi_x|\sigma'_{YZ|x} - \sigma_{YZ|x}|\Pi_x) \\
 & \geq \sfP_Y(\clE_x) - \Vert\sigma'_{YZ|x} - \sigma_{YZ|x}\Vert_1 \\
 & \geq 1-2\delta/\eps - \delta
\end{align*}
for $x\in\text{Good}_1\cap\text{Good}_2$. Hence for such $x$,
\[ \sfS_\infty(\rho'_{YZ|x}\Vert\rho_{Y|x}\otimes\theta_{Z|x}) \leq k+1 + \log\left(\frac{1}{1-2\delta/\eps-\delta}\right) \leq \frac{4c+1}{\eps^3}.\]
Also for these $x$,
\[
\Vert\rho'_{YZ|x} - \rho_{YZ|x}\Vert_1 \leq \Vert\rho'_{YZ|x} - \sigma'_{YZ|x}\Vert_1 + \Vert\sigma'_{YZ|x} - \sigma_{YZ|x}\Vert_1 + \Vert\sigma_{YZ|x} - \rho_{YZ|x}\Vert_1 \leq \eps + \frac{7\delta}{\eps}
\]
which gives us
\[ \sfI^{\eps + 7\delta/\eps}_{\max}\left(Y:Z\right)_{\rho_x} \leq \frac{4c+1}{\eps^3}.\]
We know $\sfP_X(\text{Good}_1\cap\text{Good}_2) \geq 1-2\eps$. Hence $\sfP_{X'}(\text{Good}_1\cap\text{Good}_2) \geq 1 - 2\eps - \delta/2$, which gives us the desired result.
\end{proof}

\begin{fact}[Quantum Raz's Lemma, \cite{BVY15}]
Let $\rho_{XY}$ and $\sigma_{XY}$ be two CQ states with $X = X_1\ldots X_k$ being classical, and $\sigma$ being product across all registers. Then,
\[ \sum_{i=1}^k\sfI(X_i:Y)_\rho \leq \sfS(\rho_{XY}\Vert \sigma_{XY}).\]
\end{fact}

\begin{fact}[\cite{JRS03}, Lemma 2]\label{jrs-proj}
Suppose the state
\[ \ket{\sigma}_{X\tilde{X}AB} = \sum_x\sqrt{\sfP_X(x)}\ket{xx}_{X\tilde{X}}\ket{\sigma}_{AB|x}\]
satisfies $\sfI_{\max}^\delta(X:B)_\sigma \leq k$ for some $\delta >0$. Then there is a family of measurement operators $\{\Pi_x\}_x$ acting only on $X\tilde{X}A$ such that:
\begin{enumerate}[(i)]
\item Each $\Pi_x$ succeeds with probability $\alpha = 2^{-k/\delta}$ on $\ket{\sigma}_{X\tilde{X}AB}$,
\item $(\Pi_x\otimes\Id_B)\state{\sigma}(\Pi_x\otimes\Id_B)$ is of the form $\state{xx}\otimes\rho_x$, for some state $\rho_x$ on $AB$, and
\[ \bbE_{\sfP_X}\left\Vert \frac{1}{\alpha}(\Pi_x\otimes\Id_B)\state{\sigma}_{X\tilde{X}AB}(\Pi_x\otimes\Id_B) - \state{xx}_{X\tilde{X}}\otimes\state{\sigma}_{AB|x}\right\Vert_1 \leq \delta. \footnote{The version of this fact stated here is slightly different from the original statement in \cite{JRS03}, in order to suit our application. In the original statement, $\sfI(X:B)$ is used instead of $\sfI^\delta_{\max}(X:B)$, and the superposition state lacks the $\tilde{X}$ register. However, in the proof of the fact in \cite{JRS03}, $\sfI(X:B)$ is converted to $\sfI^\delta_{\max}(X:B)$ anyway, so the first change makes no difference. The second change also makes no difference as the same projector that takes the superposition state without the $\tilde{X}$ register to $\state{x}\otimes\state{\sigma}_{AB|x}$ takes the superposition state with the $\tilde{X}$ register to $\state{xx}\otimes\state{\sigma}_{AB|x}$.}\]
\end{enumerate}
\end{fact}

\subsection{Quantum communication \& entangled games}
We briefly describe a quantum communication protocol $\clP$ for computing a relation $f\subseteq \clX\times\clY\times\clZ$, between two parties Alice and Bob sharing prior entanglement, with inputs $x$ and $y$ respectively.

In each round, either Alice or Bob will apply a unitary on their classical input register, along with the quantum register they received as a message from the other party in the last round, and memory registers they may have kept from previous rounds; after the unitary they will keep some registers as memory and send the rest to the other party as the message for that round. We can always assume that players make `safe' copies of their inputs using CNOT gates in such protocols, so that the input registers come out as is after each round. We also note that though in general we need not consider shared classical randomness in quantum communication protocols, protocols with shared randomness fall under the shared entanglement framework we have described. This is because shared randomness can be obtained by sharing entanglement and then both parties measuring in the same basis.

In a one-way, i.e., a single round protocol, the memory from previous rounds is replaced by Alice's (who we consider to be sending the single message) part of the shared entangled state, and any register she does not send as a message is simply discarded. After Alice's message, Bob performs a projective measurement on his input register, his part of the shared entanglement, and Alice's message, and gives the outcome of this measurement as the output of the protocol, which we shall denote by $\clP(x,y)$. We can of course think of this measurement as Bob performing a unitary on the three registers, and then doing a measurement in the computational basis on some $\log|\clZ|$ qubits which are designated for the output.
\begin{definition}
The one-way entanglement-assisted quantum communication complexity, with error $0<\eps<1$, of a relation $f\subseteq \clX\times\clY\times\clZ$, denoted by $\Qi_\eps(f)$, is the minimum message size, i.e., number of qubits sent, in a one-way entanglement-assisted quantum protocol $\clP$ such that for all $(x,y) \in \clX\times\clY$,
\[ \Pr[\clP(x,y) \in f(x,y)] \geq 1-\eps,\]
where the probability is taken over the inherent randomness in the protocol.
\end{definition}
\begin{definition}
For a probability distribution $p$ on $\clX\times\clY$, the distributional one-way entanglement-assisted quantum communication complexity of a relation $f\subseteq\clX\times\clY\times\clZ$, with error $0<\eps<1$ with respect to $p$, is defined as the minimum message size of a one-way entanglement-assisted quantum protocol $\clP$ such that
\[ \Pr[\clP(x,y) \in f(x,y)] \geq 1-\eps,\]
where the probability is taken over the distribution $p$ on $(x,y)$ as well as the inherent randomness in the protocol.
\end{definition}

\begin{fact}[Yao's lemma, \cite{Yao79}]\label{yao}
For any $0< \eps < 1$, and any relation $f$, $\Qi_{\eps}(f) = \max_p\Qi_{p,\eps}(f)$.
\end{fact}

A two-player non-local game $G$ is described as $(q, \clX\times\clY,\clA\times\clB, \sfV)$ where $q$ is a distribution over the input set $\clX\times\clY$, $\clA\times\clB$ is the output set, and $\sfV:\clX\times\clY\times\clA\times\clB \to \{0,1\}$ is a predicate. It is played as follows: a referee selects inputs $(x,y)$ according to $q$, sends $x$ to Alice and $y$ to Bob. If Alice and Bob are allowed to share entanglement, they perform measurements on their respective halves of the entangled state along with their respective input registers (which we model as performing unitaries and then measuring in the computational basis on some $\log|\clA|$ and $\log|\clB|$ qubits designated for outputs respectively), and send their outputs $(a,b)$ back to the referee. The referee accepts and Alice and Bob win the game iff $\sfV(x,y,a,b) = 1$.
\begin{definition}
The entangled value of a game $G=(q,\clX\times\clY,\clA\times\clB,\sfV)$, denoted by $\omega^*(G)$, is the maximum winning probability of Alice and Bob, averaged over the distribution $q$ as well as inherent randomness in the strategy, over all shared entanglement strategies for $G$.
\end{definition}

%------------------------------------------

\section{Proof of direct product theorem}
In this section, we prove Theorem \ref{thm:dpt}, whose statement we recall below.
\main*

Let $p$ be the hard distribution on $\clX \times \clY$ for $\Qi_{\eps + 12\zeta}(f)$ from Yao's lemma, i.e., $\Qi_{\eps + 12\zeta}(f) = \Qi_{p,\eps + 12\zeta}(f)$. Consider the relation $\tilde{f} \subseteq \clX\times(\clY\cup\{y^*\})\times\clZ$ which is the same as $f$ on $\clX\times\clY\times\clZ$ and additionally,
\[(x,y^*,z) \in \tilde{f} \hspace{0.2cm} \forall x \in \clX, \forall z \in \clZ.\]
We can think of $p$ as a distribution on $\clX\times(\clY\cup\{y^*\})$ as well, which has $p(y^*) = 0$. Clearly,
\begin{equation}\label{ds-lb1}
\Qi_{p, \gamma}(\tilde{f}) = \Qi_{p, \gamma}(f)
\end{equation}
for any error $\gamma$, since $p$ has no support on the extra inputs on which $\tilde{f}$ is defined. We also note that
\begin{equation}\label{ds-ubn}
\Qi_{\gamma}(f^k) \geq \Qi_{\gamma}(\tilde{f}^k)
\end{equation}
for any $\gamma$. This is because any protocol for $f^k$ is also a protocol for $\tilde{f}^k$: on the indices where Bob's input is $y^*$ instead of an element of $\clY$, he pretends he has gotten an input from $\clY$, runs the protocol with this input and gives the answer accordingly. This gives a correct output if the original protocol gives a correct output, since any output is correct when Bob's input in $y^*$.

For a distribution $q$ related to $p$, we shall show that
\begin{equation}\label{ds-main}
\Qi_{q^k, 1 - (1-\eps)^{\Omega(\zeta^6k/\log|\clZ|)}}(\tilde{f}^k) \geq \frac{\zeta^5 k}{300}\cdot\Qi_{p,\eps+12\zeta}(\tilde{f}) - k\log\log\left(\frac{24}{5\zeta}\right).
\end{equation}
Since $\Qi_{\gamma}(\tilde{f}^k) \geq \Qi_{q^k, \gamma}(\tilde{f}^k)$, \eqref{ds-lb1}, \eqref{ds-ubn} and \eqref{ds-main} imply the theorem. The distribution $q$ is defined as follows
\begin{align*}
q(x,y) & = (1-\zeta)\cdot p(x,y) \hspace{0.2cm} \forall x\in\clX, y\in\clY \\
q(x,y^*) & = \zeta \cdot p(x) \hspace{0.2cm} \forall x\in\clX.
\end{align*}
Clearly, $q(x,y^*)=q(x)q(y^*)$ for all $x$, and
\begin{equation}\label{eq:pq-dist}
\Vert p(x,y) - q(x,y)\Vert_1 \leq 2\zeta. 
\end{equation}

Following \cite{BVY15}, for each $i \in [k]$, we shall define a joint distribution $\sfP_{X_iY_iD_iG_i}$, where the marginal on $X_iY_i$ is $q(x,y)$, and $D_iG_i$ are correlation-breaking variables such that conditioned on $D_iG_i = d_ig_i$, $X_i$ and $Y_i$ are independent. Each $X_iY_iD_iG_i$ is distributed independently of the rest. Each $D_i$ is distributed uniformly in $\{0,1\}$. Depending on the value of $D_i$, $G_i$ is distributed in the following way:
\[ G_i = \left\{ \begin{array}{lll} x & \text{w.p. } p(x) & \text{if } D_i = 0 \\ y^* & \text{w.p. } 1-(1-\zeta)^{2/3} & \text{if } D_i=1 \\ y & \text{w.p. } (1-\zeta)^{2/3}\cdot p(y) & \text{if } D_i=1\end{array} \right.\]
Now depending on the value of $D_iG_i$, $X_iY_i$ is distributed in the following way:
\[ X_iY_i = \left\{\begin{array}{lll} (x,y^*) & \text{w.p. } \zeta & \text{if } D_i=0, G_i = x \\ (x,y) & \text{w.p. } (1-\zeta)\cdot p(y|x) & \text{if } D_i= 0, G_i = x \\ (x,y^*) & \text{w.p. } p(x) & \text{if } D_i=1, G_i = y^* \\ (x,y^*) & \text{w.p. } \left(1-(1-\zeta)^{1/3}\right) \cdot p(x|y) & \text{if } D_i=1, G_i = y \\ (x,y) & \text{w.p. } (1-\zeta)^{1/3}\cdot p(x|y) & \text{if } D_i=1, G_i = y.\end{array} \right.\]
The following lemma is similar to Claim 18 from \cite{BVY15}; we provide a proof for completeness.
\begin{lemma}
For all $(x,y) \in \clX\times(\clY\cup\{y^*\})$, $\sfP_{X_iY_i}(x,y) = q(x,y)$.
\end{lemma}
\begin{proof}
It is trivial to see that $\sfP_{G_iY_i|D_i=0}(x,y) = \sfP_{X_iY_i|D_i=0}(x,y) = q(x,y)$, since $G_i=X_i$ conditioned on $D_i=0$. We now prove the $D_i=1$ case. First consider a $y\in\clY$. $Y_i$ can only take value $y$ if $G_i$ takes value $y$. Hence,
\begin{align*}
\sfP_{X_iY_i|D_i=1}(x,y) & = \sfP_{G_i|D_i=1}(y)\cdot\sfP_{X_iY_i|D_i=1,G_i=y}(x,y) \\
 & = (1-\zeta)^{2/3}p(y)\cdot(1-\zeta)^{1/3}p(x|y) \\
 & = (1-\zeta)\cdot p(x,y) = q(x,y).
\end{align*}
On the other hand, $Y_i$ can take value $y^*$ when $G_i=y^*$ or when $G_i=y$ for any $y\in\clY$. Hence,
\begin{align*}
\sfP_{X_iY_i|D_i=1}(x,y^*) & = \sfP_{G_i|D_i=1}(y^*)\cdot\sfP_{X_iY_i|D_i=1,G_i=y^*}(x,y^*) + \sum_{y\in\clY}\sfP_{G_i|D_i=1}(y)\cdot\sfP_{X_iY_i|D_i=1,G_i=y}(x,y^*) \\
 & = \left(1-(1-\zeta)^{2/3}\right)\cdot p(x) + (1-\zeta)^{2/3}\left(1 - (1-\zeta)^{1/3}\right)\sum_{y\in \clY}p(y)\cdot p(x|y) \\
 & = \left(1-(1-\zeta)^{2/3}\right)\cdot p(x) + \left((1-\zeta)^{2/3} - (1-\zeta)\right)\cdot p(x) \\
 & = \zeta\cdot p(x) = q(x,y^*). \qedhere
\end{align*}
\end{proof}
In particular the lemma means $\sfP_{X_iY_i}(x,y^*) = \sfP_{X_i}(x)\sfP_{Y_i}(y^*)$. We also note
\begin{equation}\label{eq:G=Y}
\sfP_{Y_iG_i|D_i=1}(Y_i\neq G_i) = (1-\zeta)^{2/3}(1-(1-\zeta)^{1/3}) \leq 1-2\zeta/3 - 1 + \zeta = \zeta/3.
\end{equation}

To prove \eqref{ds-main}, let $\clP$ be any quantum one-way protocol between Alice and Bob, for $\tilde{f}^k \subseteq \clX^k \times (\clY\cup\{y^*\})^k\times \clZ^k$. $\clP$ is depicted in Figure \ref{fig:prot}. Alice and Bob's inputs are in registers $X=X_1\ldots X_k$ and $Y=Y_1\ldots Y_k$, and they share an entangled pure state uncorrelated with the inputs on registers $E_AE_B$, with Alice holding $E_A$ and Bob holding $E_B$. Alice applies a unitary $V^\text{Alice}$ on $XE_A$, to get the message register $M$, and the register $A$ to be discarded. We shall use $\ket{\theta}_{AME_B|x}$ to refer to the pure state in $AME_B$ in the protocol after Alice's unitary, for inputs $xy$ ($\ket{\theta}_x$ only depends on $y$ via $x$). When Alice and Bob's inputs are distributed according to $\sfP_{XY}$, the state of the protocol after Alice's message, will be given by the following CQ state:
\[ \theta_{XYAME_B} = \sum_{xy}\sfP_{XY}(xy)\state{xy}_{XY}\otimes\state{\theta}_{AME_B|x}.\]
We shall also consider the following purification of it, with the purifying registers $\tilde{X}$ and $\tilde{Y}$:
\[ \ket{\theta}_{X\tilde{X}Y\tilde{Y}AME_B} = \sum_{xy}\sqrt{\sfP_{XY}(xy)}\ket{xxyy}_{X\tilde{X}Y\tilde{Y}}\ket{\theta}_{AME_B|x}.\]
After receiving Alice's message, Bob applies a unitary $V^\text{Bob}$ to $YME_B$, after which $ME_B$ gets converted to $BZ$, where $Z=Z_1\ldots Z_k$ are the answer registers. We shall use $\ket{\rho}_{X\tilde{X}Y\tilde{Y}ABZ}$ to refer to $\ket{\theta}_{X\tilde{X}Y\tilde{Y}AME_B}$ after $V^\text{Bob}$. We shall use $\sfP_{XYDGZ}$ to refer to the joint distribution of these variables in $\ket{\rho}$, where the $Z$ distribution is obtained by measuring the $Z$ register in the computational basis.

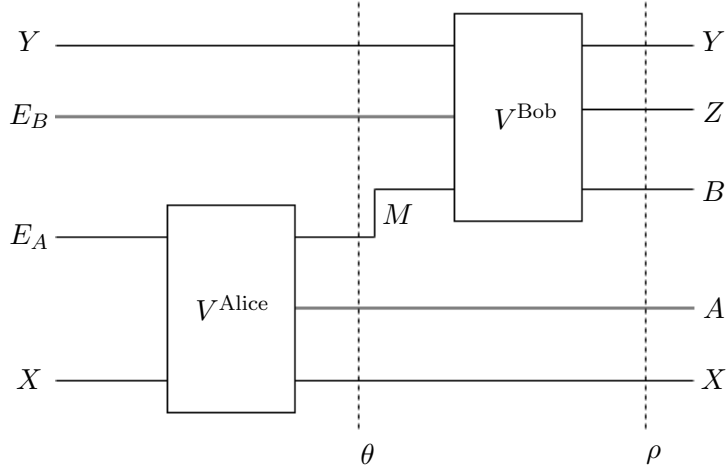
\begin{figure}[!h]
\centering
\begin{tikzpicture}
\node at (0,0) {\includegraphics[scale=0.6]{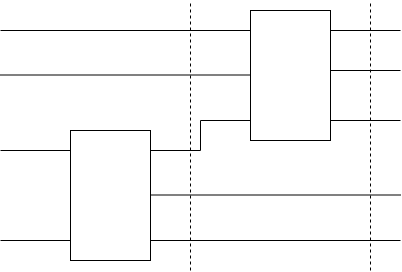}};
\node at (-4.6,-2.2) {$X$};
\node at (-4.6, -0.3) {$E_A$};
\node at (-4.6, 1.3) {$E_B$};
\node at (-4.6, 2.3) {$Y$};
\node at (-1.9,-1.25) {$V^\text{Alice}$};
\node at (0.3, 0) {$M$};
\node at (2, 1.3) {$V^\text{Bob}$};
\node at (4.5,-2.2) {$X$};
\node at (4.5, -1.25) {$A$};
\node at (4.5, 0.3) {$B$};
\node at (4.5, 1.35) {$Z$};
\node at (4.5, 2.3) {$Y$};
\node at (-0.1, -3.2) {$\theta$};
\node at (3.7, -3.2) {$\rho$};
\end{tikzpicture}
\caption{One-way quantum protocol $\clP$}
\label{fig:prot}
\end{figure}

We shall show that if the communication cost of $\clP$ is $< \frac{\zeta^5 k}{300}\cdot\Qi_{p, \eps + 12\zeta}(\tilde{f}) - k\log\log(24/5\zeta)$, then the success probability of $\clP$ is $(1-\eps)^{\Omega(\zeta^6k/\log|\clZ|)}$. This is implied by the following claim, which the rest of the proof will show. \qedhere

\begin{lemma}\label{ind}
Let $\delta = \frac{\zeta^6}{1440000}$ and $\delta' = \frac{\zeta^6}{1440000\log|\clZ|}$. For $i \in [k]$, let $T_i$ be the random variable which takes value 1 if $\clP$ computes $f(X_i,Y_i)$ correctly, and value 0 otherwise. If the communication cost of $\clP$ is $< \frac{\zeta^5 k}{300}\cdot\Qi_{p, \eps + 12\zeta}(\tilde{f}) - k\log\log(24/5\zeta)$, then there exist $\lfloor \delta' k\rfloor$ coordinates $\{i_1, \ldots, i_{\lfloor \delta' k \rfloor}\} \subseteq [k]$, such that for all $1 \leq r \leq \lfloor \delta' k\rfloor - 1$, at least one of the following two conditions holds
\begin{enumerate}[(i)]
\item $\Pr\left[\prod_{j =1}^rT_{i_j} = 1\right] \leq (1-\eps)^{\delta k}$
\item $\Pr\left[T_{i_{r+1}} = 1 \middle| \prod_{j=1}^rT_{i_j} = 1\right] \leq 1-\eps$.
\end{enumerate}
\end{lemma}
Lemma \ref{ind} can be proved inductively. Suppose we have already identified $1 \leq t \leq \lfloor \delta' k\rfloor$ coordinates in $C = \{i_1, \ldots i_t\}$, such that for all $1\leq r \leq t-1$, $\Pr\left[T_{i_{r+1}} = 1| \prod_{j=1}^r T_{i_j} = 1\right] \leq 1-\eps$. Let $\clE$ refer to the event $\prod_{i \in C}T_i = 1$. If $\Pr[\clE] \leq (1-\eps)^{\delta k}$, then we are already done. If not, then we shall show how to identify the $(t+1)$-th coordinate $i$ such that $\Pr\left[T_i = 1 |\clE\right] \leq 1 - \eps$. The process of identifying the first coordinate is also similar, except in that case the conditioning event is empty. Since we only use the lower bound $(1-\eps)^{\delta k}$ on the probability of the conditioning event in our proof, the proof goes through for that case as well.

We shall use the state $\ket{\vph}$, which is $\ket{\rho}_{X\tilde{X}Y\tilde{Y}ABZ}$ conditioned on $\clE$, for the proof of Lemma \ref{ind}. For any value $DG=dg$, $\ket{\vph}_{X\tilde{X}Y\tilde{Y}ABZ|dg}$ is defined as:
\[
\ket{\vph}_{X\tilde{X}Y\tilde{Y}ABZ|dg} 
= \frac{1}{\sqrt{\gamma_{dg}}}\sum_{xy}\sqrt{\sfP_{XY|dg}(xy)}\ket{xxyy}_{X\tilde{X}Y\tilde{Y}}\otimes\sum_{z_C:(x_C,y_C,z_C) \in \tilde{f}^t}\ket{z_C}_{Z_C}\ket{\tilde{\vph}}_{ABZ_{\bar{C}}|{xyz_C}}.
\]
Here $\ket{\tilde{\vph}}_{xyz_C}$ is a subnormalized state with $
\Vert\ket{\tilde{\vph}}_{ABZ_{\bar{C}}|xyz_C}\Vert_2^2 = \sfP_{Z_C|xy}(z_C)$. The overall normalization factor $\gamma_{dg}$ is the probability of $\clE$ conditioned on $dg$, and satisfies
\[ \sum_{dg}\sfP_{DG}(dg)\cdot\gamma_{dg} = \Pr[\clE].\]
It is clear that the distribution of $XYZ$ in $\ket{\vph}_{X\tilde{X}Y\tilde{Y}ABZ|dg}$ is $\sfP_{XYZ|\clE,dg}$. Note that we are using the notation $\ket{\vph}_{dg}$ without explicitly considering registers $DG$ on which a measurement is done to obtain $\ket{\vph}_{dg}$. We shall also sometimes use $\ket{\vph}_{d_{-i}g_{-i}}$ in which the $xy$ distributions are conditioned on $d_{-i}g_{-i}$ instead, which changes the normalization factor to some $\gamma_{d_{-i}g_{-i}}$, everything else remaining the same. $\vph_{x_iy_id_{-i}g_{-i}}$ refers as usual to the state obtained when a measurement done on the $X_iY_i$ registers (which are actually present in $\ket{\vph}$) in $\ket{\vph}_{d_{-i}g_{-i}}$. For $i\notin\bar{C}$, we shall use the states $\ket{\vph}_{X_{\bar{C}}\tilde{X}_{\bar{C}}Y_{\bar{C}}\tilde{Y}_{\bar{C}}ABZ_{\bar{C}}|x_iy_ix_Cy_Cz_Cd_{-i}g_{-i}}$ in our proof, which we note are pure states.

\vspace{0.2cm}
Lemma \ref{ind} will be proved with the help of the following lemma, whose proof we give later.
\begin{lemma}\label{i-conds}
If $\Pr[\clE] \geq (1-\eps)^{\delta k}$, then there exist a coordinate $i \in \bar{C}$, a random variable $R_i = X_CY_CZ_CD_{-i}G_{-i}$ and for each $R_i=r_i$ a state $\ket{\vph'}_{X_{\bar{C}}\tilde{X}_{\bar{C}}Y_{\bar{C}}\tilde{Y}_{\bar{C}}ABZ_{\bar{C}}|y^*r_i}$ such that the following conditions hold:
\begin{enumerate}[(i)]
\item $\Vert\sfP_{X_iY_iR_i|\clE} - \sfP_{X_iY_i}\sfP_{R_i|\clE,X_i}\Vert_1 \leq \frac{7\zeta}{120}$ \label{eq:XY(R|X)i}
\item $\Vert\sfP_{X_iY_iR_i|\clE} - \sfP_{X_iY_i}\sfP_{R_i|\clE,Y_i}\Vert_1 \leq \frac{7\zeta}{120}$. \label{eq:XY(R|Y)i}
\end{enumerate}
There exist projectors $\{\Pi_{x_ir_i}\}_{x_ir_i}$ acting only on registers $X_{\bar{C}}\tilde{X}_{\bar{C}}A$ and unitaries $\{U_{y_ir_i}\}_{y_ir_i}$ acting only on $Y_{\bar{C}}\tilde{Y}_{\bar{C}}BZ_{\bar{C}}$, such that each $\Pi_{x_ir_i}$ succeeds on $\ket{\vph'}_{r_i}$ with probability $\alpha_{r_i} = 2^{-c'_{r_i}}$, and 
\begin{enumerate}[(i)]
\setcounter{enumi}{2}
\item $\bbE_{\sfP_{R_i|\clE}}c'_{r_i} \leq \frac{300c}{\zeta^5}$ \label{eq:X-proj-succ-i}
\item $\bbE_{\sfP_{X_iY_iR_i|\clE}}\left\Vert\frac{1}{\alpha_{r_i}}(\Pi_{x_ir_i}\otimes U_{y_ir_i})\state{\vph'}_{y^*r_i}(\Pi_{x_ir_i}\otimes U^\dagger_{y_ir_i}) - \state{\vph}_{x_iy_ir_i}\right\Vert_1 \leq 21\zeta.$ \label{eq:XY-dist-i}
\end{enumerate}
\end{lemma}

\begin{proof}[Proof of Lemma \ref{ind}]
We give a one-way quantum protocol $\clP'$ for $\tilde{f}$, whose inputs are distributed according to $\sfP_{X_iY_i}$, i.e., $q$, by embedding Alice and Bob's inputs into the $i$-th coordinate of $\ket{\vph}_{x_iy_ir_i}$, as follows:
\begin{itemize}
\item Alice and Bob have $r$ according to the distribution required by Fact \ref{embed} as shared randomness, and $2^{300c/\zeta^5}\log(24/5\zeta)$ copies of $\ket{\vph'}_{y^*r_i}$ as shared entanglement, with Alice holding registers $X_{\bar{C}}\tilde{X}_{\bar{C}}A$ and Bob holding registers $Y_{\bar{C}}\tilde{Y}_{\bar{C}}BZ_{\bar{C}}$ of each copy.
\item On input $(x_i,y_i)$ from $\sfP_{X_iY_i}$, using items \eqref{eq:XY(R|X)i}, \eqref{eq:XY(R|Y)i} of Lemma \ref{i-conds}, their shared randomness, and the protocol from Fact \ref{embed}, Alice and Bob generate random variables $R^\text{Alice}_iR^\text{Bob}_i$ such that
\[ \Vert\sfP_{X_iY_iR^\text{Alice}_iR^\text{Bob}_i} - \sfP_{X_iY_iR_iR_i|\clE}\Vert_1 \leq \frac{7\zeta}{24}.\]
where $R_iR_i$ denotes two perfectly correlated copies of $R_i$ in $\sfP_{X_iY_iR_iR_i|\clE}$.
\item Alice applies the $\{\Pi_{x_ir^\text{A}_i}, \Id-\Pi_{x_ir^\text{A}_i}\}$ measurement according to her input and $R^\text{Alice}_i$ on her registers for each copy of the shared entangled state. If the $\Pi_{x_ir^\text{A}_i}$ measurement does not succeed on any copy, then she aborts. Otherwise, she sends to Bob a $(\frac{300c}{\zeta^5} + \log\log(24/5\zeta))$-bit message indicating an index where $\Pi_{x_ir^\text{A}_i}$ measurement succeeded.
\item Bob applies the unitary $U_{y_ir^\text{B}_i}$ according to his input and $R^\text{Bob}_i$ on the copy of the shared entangled state whose index Alice has sent, and measures the $Z_i$ register of the resulting state to give her output.
\end{itemize}
To analyze the success of this protocol, first note that
\[ \bbE_{\sfP_{X_iY_iR_i|\clE}}\Pr[\text{Result of $Z_i$ measurement on } \ket{\vph}_{x_iy_ir_i} \in \tilde{f}(x_i,y_i)] = \Pr[T_i=1|\clE].\]
Let us first assume Alice and Bob have $(x_i,y_i,r^\text{A}_i, r^\text{B}_i)$ distributed exactly according to $\sfP_{X_iY_iR_iR_i|\clE}$ -- we shall denote both $r^\text{A}_i$ and $r^\text{B}_i$ by $r_i$ in this case. Alice aborts the protocol if none of her measurements succeed. On expectation, this happens with probability
\[ \bbE_{\sfP_{R_i|\clE}}\left[(1-2^{-c'_{r_i}})^{300c/\zeta^5\log(24/7\zeta)}\right] \leq \left(1-2^{-\bbE_{\sfP_{R_i|\clE}} c'_{r_i}}\right)^{300c/\zeta^5\log(24/5\zeta)} \leq \frac{5\zeta}{24} \]
from \eqref{eq:X-proj-succ-i}. If Alice does not abort, then Alice and Bob's state after Bob's unitary is $\frac{1}{\sqrt{\alpha_{r_i}}}(\Pi_{x_ir_i}\otimes U_{y_ir_i}\ket{\vph'}_{y^*r_i}$. From \eqref{eq:XY-dist-i}, the expected probability of the $Z_i$ measurement on this state giving an answer $\in \tilde{f}(x_i,y_i)$ is at least $\Pr[T_i=1|\clE] - \frac{21\zeta}{2}$. Hence, if Alice and Bob had $(x_i,y_i,r^\text{A}_i,r^\text{B}_i)$ distributed according to $\sfP_{X_iY_iR_iR_i|\clE}$, then their expected success probability would have been at least $\Pr[T_i=1|\clE]- \frac{21\zeta}{2}-\frac{5\zeta}{24}$. Since Alice and Bob have $(x_i,y_i,r^\text{A}_i,r^\text{B}_i)$ according to $\sfP_{X_iY_iR^\text{Alice}_iR^\text{Bob}_i}$ instead, their expected success probability is at least
\[ \Pr[T_i=1|\clE] - \frac{21\zeta}{2} - \frac{5\zeta}{24} - \frac{7\zeta}{24} \geq \Pr[T_i=1|\clE] - 11\zeta.\]
Since $\Vert q(x,y) - p(x,y)\Vert_1 \leq 2\zeta$, when the same protocol is run on $X_iY_i$ distributed according to $p$ instead, it must succeed with probability at least $\Pr[T_i=1|\clE] - 12\zeta$. Since the communication in $\clP'$ is at most $(\frac{300c}{\zeta^5} + \log\log(24/5\zeta)) < \Qi_{p,\eps + 12\zeta}(\tilde{f})$, $\Pr[T_i=1|\clE] \geq 1-\eps$ gives the error probability of $\clP'$ to be $\leq \eps + 12\zeta$, which is a contradiction. Hence we must have $\Pr[T_i=1|\clE] \leq 1-\eps$. The desired result thus follows by setting $i_{t+1}=i$.
\end{proof}

\begin{proof}[Proof of Lemma \ref{i-conds}]
Applying Fact \ref{hol} with $T$ and $V$ being trivial and $U_i = X_iY_iD_iG_i$ for $i \in \bar{C}$, we get,
\begin{equation}
\bbE_{i \in \bar{C}}\Vert\sfP_{X_iY_iD_iG_i|\clE} - \sfP_{X_iY_iD_iG_i}\Vert_1 \leq \frac{1}{k-t}\sqrt{k\cdot \log((1-\eps)^{\delta k})} \leq \sqrt{2\delta}.
\end{equation}
In particular, due to \eqref{eq:G=Y}, this means
\begin{equation}\label{eq:G=Y-2}
\bbE_{i\in\bar{C}}\sfP_{Y_iG_i|\clE,D_i=1}(Y_i=G_i) \geq 1-\zeta/3-\sqrt{2\delta}.
\end{equation}
And since $\sfP_{G_i|D_i=1}(y^*) = 1 - (1-\zeta)^{2/3}$, $\sfP_{Y_i|D_i=1,G_i=y}(y_i) = (1-\zeta)^{1/3}$ for $y_i\in\clY$, we have
\begin{equation}\label{eq:prob-y*}
 \zeta+\sqrt{2\delta} \geq 1- (1-\zeta)^{2/3} + \sqrt{2\delta} \geq \bbE_{i\in\bar{C}}\sfP_{G_i|\clE,D_i=1}(y^*) \geq 1 - (1-\zeta)^{2/3} - \sqrt{2\delta} \geq 2\zeta/3 - \sqrt{2\delta}
\end{equation}
\begin{equation}\label{eq:Y=G=y}
(1-\zeta/3 + \sqrt{2\delta})\cdot\bbE_{i\in\bar{C}}\sfP_{G_i|\clE,D_i=1}(y_i) \geq \bbE_{i\in\bar{C}}\sfP_{Y_iG_i|\clE,D_i=1}(y_i,y_i) \geq (1-\zeta-\sqrt{2\delta})\cdot\bbE_{i\in\bar{C}}\sfP_{G_i|\clE,D_i=1}(y_i).
\end{equation}
Fact \ref{hol} can again be applied with $U_i = X_iY_i$, $T=X_CY_CDG$ and $V=Z_C$. Let $\delta_1 = \delta + \delta'\log|\clZ| = \frac{\zeta^6}{720000}$. Then we have,
\begin{align}
\sqrt{2\delta_1} & \geq \bbE_{i\in\bar{C}}\Vert\sfP_{X_iY_iX_CY_CZ_CDG|\clE} - \sfP_{X_CY_CZ_CDG|\clE}\sfP_{X_iY_i|X_CY_CDG}\Vert_1 \nonumber \\
 & = \bbE_{i\in\bar{C}}\Vert\sfP_{X_iY_iX_CY_CZ_CDG|\clE} - \sfP_{X_CY_CZ_CDG|\clE}\sfP_{X_iY_i|D_iG_i}\Vert_1 \nonumber \\
 & = \bbE_{i\in\bar{C}}\Vert\sfP_{X_iY_iD_iG_iR_i|\clE} - \sfP_{D_iG_iR_i|\clE}\sfP_{X_iY_i|D_iG_i}\Vert_1. \label{eq:XYR-1}
\end{align}
 We note that $D_i$ takes value uniformly in $\{0,1\}$ even conditioned on $\clE$. Hence from \eqref{eq:XYR-1},
\begin{align*}
\sqrt{2\delta_1} & \geq \frac{1}{2}\bbE_{i\in\bar{C}}\Vert\sfP_{X_iY_iG_iR_i|\clE,D_i=0} - \sfP_{G_iR_i|\clE,D_i=0}\sfP_{X_iY_i|G_i,D_i=0}\Vert_1 \\
& = \frac{1}{2}\bbE_{i\in\bar{C}}\Vert\sfP_{X_iY_iR_i|\clE} - \sfP_{X_iR_i|\clE}\sfP_{Y_i|X_i}\Vert_1 
\end{align*}
where we have used the fact that $X_i=G_i$ conditioned on $D_i=0$. Combining this with the fact that $\bbE_{i\in\bar{C}}\Vert \sfP_{X_i|\clE} - \sfP_{X_i}\Vert_1 \leq \sqrt{2\delta}$, we have,
\begin{equation}\label{eq:XY(R|X)}
\bbE_{i\in\bar{C}}\Vert\sfP_{X_iY_iR_i|\clE} - \sfP_{X_iY_i}\sfP_{R_i|\clE,X_i}\Vert_1 \leq 3\sqrt{2\delta_1} < \frac{7\zeta^3}{600}.
\end{equation}
Due to Corollary \ref{x*y*} we also have from \eqref{eq:XY(R|X)},
\begin{equation}\label{eq:XRy*}
\bbE_{i\in\bar{C}}\Vert\sfP_{X_iR_i|\clE, y^*} - \sfP_{X_iR_i|\clE}\Vert_1 \leq \frac{33\sqrt{2\delta_1}}{\zeta}.
\end{equation}
Let $\clF_i$ denote the event $Y_i=G_i$. We know $\bbE_{i\in\bar{C}}\sfP_{X_iY_iG_i|D_i=1}(\clF_i) \geq 1-\zeta/3 - \sqrt{2\delta}$, from \eqref{eq:G=Y-2}. Hence, using Fact \ref{cond-prob},
\begin{align*}
\bbE_{i\in\clC}\Vert\sfP_{X_iY_iR_i|\clE} - \sfP_{Y_iR_i|\clE}\sfP_{X_i|Y_i}\Vert_1 & = \bbE_{i\in\bar{C}}\Vert\sfP_{X_iY_iG_iR_i|\clE,D_i=1,\clF_i} - \sfP_{G_iR_i|\clE,D_i=1,\clF_i}\sfP_{X_iY_i|G_iD_i=1,\clF_i}\Vert_1 \\
 & \leq 6\bbE_{i\in\bar{C}}\Vert\sfP_{X_iY_iD_iR_i|\clE,G_i=1} - \sfP_{G_iR_i|\clE,D_i=1}\sfP_{X_iY_i|G_iD_i=1}\Vert_1 \leq 6\sqrt{2\delta_1}.
\end{align*}
Using $\bbE_{i\in\bar{C}}\Vert\sfP_{Y_i|\clE} - \sfP_{Y_i}\Vert_1 \leq \sqrt{2\delta}$, we have as before,
\begin{equation}\label{eq:XY(R|Y)}
\bbE_{i\in\bar{C}}\Vert\sfP_{X_iY_iR_i|\clE} - \sfP_{X_iY_i}\sfP_{R_i|\clE,Y_i}\Vert_1 \leq 7\sqrt{2\delta_1} = \frac{7\zeta^3}{600}.
\end{equation}

Let $M$ be $ck$ qubits. By Fact \ref{dim-ub}, for any value $DG=dg$, there exists some state $\sigma_{M|dg}$ such that
\[ \sfS_\infty(\theta_{XY\tilde{Y}E_BM|dg}\Vert\theta_{XY\tilde{Y}E_B|dg}\otimes\sigma_{M|dg}) \leq 2ck.\]
By Fact \ref{u-inv} we have,
\[ \sfS_\infty\left(\rho_{XY\tilde{Y}BZ|dg}\Vert V^\text{Bob}(\theta_{XY\tilde{Y}E_B|dg}\otimes\sigma_{M|dg})(V^\text{Bob})^\dagger\right) \leq 2ck.\]
Let $\psi_{X_{\bar{C}}Y_{\bar{C}}\tilde{Y}_{\bar{C}}BZ_{\bar{C}}|dg} = \Tr_{Z_{C}}(V^\text{Bob}(\theta_{XYE_B|dg}\otimes\sigma_{M|x_Cy_Cdg})(V^\text{Bob})^\dagger)$. Note that $\theta_{XY\tilde{Y}E_B|dg}\otimes\sigma_{M|dg}$ is product across $X$ and the other registers, and $V^\text{Bob}$ does not act on $X$. Hence $\psi_{X_{\bar{C}}Y_{\bar{C}}\tilde{Y}_{\bar{C}}BZ_{\bar{C}}|dg}$ is also product across $X$ and the other registers, and moreover, all the $X_i$-s are in product with each other as well. We have,
\[ \sfS_\infty\left(\rho_{XY\tilde{Y}BZ_{\bar{C}}|dg}\Vert\psi_{XY\tilde{Y}BZ_{\bar{C}}|dg}\right) \leq 2ck.\]
Using Facts \ref{cond-dec} and \ref{Sinfty-tri}, this gives us
\begin{align*}
& \bbE_{\sfP_{X_CY_CZ_CDG|\clE}}\left[\sfS\left(\vph_{X_{\bar{C}}Y_{\bar{C}}\tilde{Y}_{\bar{C}}BZ_{\bar{C}}|x_Cy_Cz_Cdg}\Vert\psi_{X_{\bar{C}}Y_{\bar{C}}\tilde{Y}_{\bar{C}}BZ_{\bar{C}}|x_Cy_Cdg}\right)\right] \\
& \leq \bbE_{\sfP_{Z_CDG|\clE}}\left[\sfS\left(\vph_{XY\tilde{Y}BZ_{\bar{C}}|z_Cdg}\Vert\psi_{XY\tilde{Y}BZ_{\bar{C}}|dg}\right)\right] \\
& \leq \bbE_{\sfP_{Z_CDG|\clE}}\left[\sfS_\infty\left(\vph_{XY\tilde{Y}BZ_{\bar{C}}|z_Cdg}\Vert\psi_{XY\tilde{Y}BZ_{\bar{C}}|dg}\right)\right] \\
& \leq \bbE_{\sfP_{Z_CDG|\clE}}\left[\sfS_\infty\left(\vph_{XY\tilde{Y}BZ_{\bar{C}}|z_Cdg}\Vert\vph_{XY\tilde{Y}BZ_{\bar{C}}|dg}\right)\right. \\
& \quad + \sfS_\infty\left(\vph_{XY\tilde{Y}BZ_{\bar{C}}|dg}\Vert\rho_{XY\tilde{Y}BZ_{\bar{C}}|dg}\right) + \left. \sfS_\infty\left(\rho_{XY\tilde{Y}BZ_{\bar{C}}|dg}\Vert\psi_{XY\tilde{Y}BZ_{\bar{C}}|dg}\right) \right] \\
& \leq \bbE_{\sfP_{Z_CDG|\clE}}\left[\log(1/\sfP_{Z_C|\clE}(z_C)) + \log(1/\Pr[\clE]) + 2ck\right] \\
& \leq |C|\log|\clZ| + \delta k + 2ck \leq (\delta_1 + 2c)k.
\end{align*}
By Quantum Raz's Lemma,
\begin{align}
4c + 2\delta_1 & \geq \bbE_{i\in\bar{C}}\bbE_{\sfP_{X_CY_CZ_CDG|\clE}}
\sfI(X_i:Y_{\bar{C}}\tilde{Y}_{\bar{C}}BZ_{\bar{C}})_{\vph_{x_Cy_Cz_Cdg}} \nonumber \\
 & = \bbE_{i\in\bar{C}}\bbE_{\sfP_{D_iG_iR_i|\clE}}\sfI(X_i:Y_{\bar{C}}\tilde{Y}_{\bar{C}}BZ_{\bar{C}})_{\vph_{d_ig_ir_i}} \nonumber \\
 & \geq \bbE_{i\in\bar{C}}\frac{1}{2}\sfP_{G_i|\clE,D_i=1}(y^*)\bbE_{\sfP_{R_i|\clE,D_i=1,G_i=y^*}}\sfI(X_i:Y_{\bar{C}}\tilde{Y}_{\bar{C}}BZ_{\bar{C}})_{\vph_{r_i|D_i=1,G_i=y^*}} \nonumber \\
 & \geq \bbE_{i\in\bar{C}}\frac{1}{2}(2\zeta/3 - \sqrt{2\delta})\bbE_{\sfP_{R_i|\clE,D_i=1,G_i=y^*}}\sfI(X_i:Y_{\bar{C}}\tilde{Y}_{\bar{C}}BZ_{\bar{C}})_{\vph_{r_i,D_i=1,G_i=y^*}} \label{eq:X-info}
\end{align}
where we have used \eqref{eq:prob-y*} in the last inequality.

Note that $\vph_{X_{\bar{C}}\tilde{X}_{\bar{C}}Y_{\bar{C}}\tilde{Y}_{\bar{C}}ABZ_{\bar{C}}|x_ir_i,D_i=1,G_i=y^*}$ is the same state as $\vph_{X_{\bar{C}}\tilde{X}_{\bar{C}}Y_{\bar{C}}\tilde{Y}_{\bar{C}}ABZ_{\bar{C}}|x_iy^*r_i}$, where the value of $Y_i$ is being conditioned on, instead of $G_i$. $\ket{\vph}_{r_i,D_i=1,G_i=y^*}$ is the superposition over $X_i$ of $\ket{\vph}_{x_ir_i,D_i=1,G_i=y^*}$, with the $X_i$ distribution being $\sfP_{X_i|\clE,r_i,D_i=1,G_i=y^*}$. The only difference between $\ket{\vph}_{y^*r_i}$ and $\ket{\vph}_{r_i,D_i=1,G_i=y^*}$ is the $X_i$ distribution, which in the former is $\sfP_{X_i|\clE,y^*r_i}$ instead. We shall refer to $\ket{\vph}_{r_i,D_i=1,G_i=y^*}$ as simply $\ket{\vph}_{r_i,1,y^*}$ as now on -- note that there is no ambiguity between this and $\ket{\vph}_{y^*r_i}$. The same goes for the distributions $\sfP_{X_iR_i|\clE,1,y^*}$ and $\sfP_{X_iR_i|\clE,y^*}$.

$\sfP_{X_i|1,y^*}$ is the same distribution as $\sfP_{X_i|y^*}$ and $\sfP_{R_i|\clE,x_i,1,y^*}$ is the same distribution as $\sfP_{R_i|\clE,x_iy^*}$ for any $x_i$. Hence,
\begin{align*}
\bbE_{i\in\bar{C}}\Vert\sfP_{X_iR_i|\clE,y^*} - \sfP_{X_iR_i|\clE,1,y^*}\Vert_1 & \leq \bbE_{i\in\bar{C}}\left[\Vert\sfP_{X_iR_i|\clE,y^*} - \sfP_{X_i|y^*}\sfP_{R_i|\clE,X_i,y^*}\Vert_1 \right. \\
 & \quad \left. + \Vert(\sfP_{X_i|1,y^*}-\sfP_{X_i|\clE,1,y^*})\sfP_{R_i|\clE,X_i,y^*}\Vert_1\right] \\
 & \leq \bbE_{i\in\bar{C}}\left[\frac{\Vert\sfP_{X_iR_i|\clE} - \sfP_{X_i}\sfP_{R_i|\clE,X_i}\Vert_1}{2\zeta/3-\sqrt{2\delta}} + \frac{\Vert\sfP_{X_i|\clE} - \sfP_{X_i}\Vert_1}{2\zeta/3 - \sqrt{2\delta}}\right] \\
 & \leq \frac{7\sqrt{2\delta_1}}{\zeta}
\end{align*}
where we have used \eqref{eq:prob-y*} in the second inequality. Using the above computation and \eqref{eq:XRy*}, we get,
\[ \bbE_{i\in\bar{C}}\Vert \sfP_{X_iR_i|\clE} - \sfP_{X_iR_i|\clE,1,y^*} \Vert_1 \leq \frac{40\sqrt{2\delta_1}}{\zeta}. \]
Let $\ket{\vph'}_{y^*r_i}$ denote the pure state where the distribution of $X_i$ is unconditioned on $Y_i=y^*$, but everything else is conditioned.  From \eqref{eq:X-info} and Fact \ref{Imax-close}, we then have that,
\[ \bbE_{i\in\bar{C}}\sfP_{R_i|\clE}\left(\sfI^{\zeta+280\sqrt{2\delta_1}/\zeta^2}_{\max}(X_i:Y_{\bar{C}}\tilde{Y}_{\bar{C}BZ_{\bar{C}}})_{\vph'_{y^*r_i}} > \frac{28(2c+\delta_1)+1}{\zeta^4}\right) \leq 2\zeta +  \frac{20\sqrt{2\delta_1}}{\zeta}.\]
Hence by Fact \ref{jrs-proj}, there exist projectors $\Pi_{x_ir_i}$ acting on registers $X_{\bar{C}}\tilde{X}_{\bar{C}}A$, such that $\Pi_{x_ir_i}$ succeeds with probability $\alpha_{r_i} = 2^{-c'_{r_i}}$ on $\ket{\vph'}_{X_{\bar{C}}\tilde{X}_{\bar{C}}Y_{\bar{C}}\tilde{Y}_{\bar{C}}ABZ_{\bar{C}}|y^*r_i}$, where
\begin{align}
\bbE_{i\in\bar{C}}\bbE_{\sfP_{R_i|\clE}}c'_{r_i} & \leq \frac{1}{\zeta}\cdot \frac{28(2c+\delta_1)+1}{\zeta^4} \leq \frac{60c}{\zeta^5} \label{eq:X-proj-succ} \\
\bbE_{\in\bar{C}}\bbE_{\sfP_{X_iR_i|\clE}}\left\Vert\frac{1}{\alpha_{r_i}}(\Pi_{x_ir_i}\otimes\Id)\state{\vph'}_{y^*r_i}(\Pi_{x_ir_i}\otimes\Id) - \state{\vph}_{x_iy^*r_i}\right\Vert_1 & \leq 3\zeta + \frac{300\sqrt{2\delta_1}}{\zeta^2} \leq \frac{7\zeta}{2}. \label{eq:X-proj-dist} 
\end{align}

By similar arguments as the ones leading to \eqref{eq:X-info} on Bob's side (except the first step where we consider the information due to the message sent by Alice to Bob, which does not apply here), we can alo upper bound $\bbE_{\sfP_{X_CY_CZ_CDG|\clE}}\left[\sfS\left(\vph_{Y_{\bar{C}}X_{\bar{C}}\tilde{X}_{\bar{C}}A|x_Cy_Cz_Cdg}\Vert\rho_{Y_{\bar{C}}X_{\bar{C}}\tilde{X}_{\bar{C}}A|x_Cy_Cdg}\right)\right]$. Hence by Raz's lemma again,
\begin{align*}
2\delta_1 & \geq \bbE_{i\in\bar{C}}\bbE_{\sfP_{D_iG_iR_i|\clE}}\sfI(Y_i:X_{\bar{C}}\tilde{X}_{\bar{C}}A)_{\vph_{d_ig_ir_i}} \\
 & \geq \bbE_{i\in\bar{C}}\frac{1}{2}(1-\zeta-\sqrt{2\delta})\bbE_{\sfP_{R_iG_i|\clE,D_i=1,G_i\neq y^*}}\sfI(Y_i:X_{\bar{C}}\tilde{X}_{\bar{C}}A)_{\vph_{r_i,D_i=1,g_i}} \\
 & = \bbE_{i\in\bar{C}}\frac{1}{2}(1-\zeta-\sqrt{2\delta})\bbE_{\sfP_{R_iG_iY_i|\clE,D_i=1,G_i\neq y^*}}\left[\sfS\left(\vph_{X_{\bar{C}}\tilde{X}_{\bar{C}}A|y_i,D_i=1,g_i}\Vert\vph_{X_{\bar{C}}\tilde{X}_{\bar{C}}A|D_i=1,g_i}\right)\right] \\
 & \geq \bbE_{i\in\bar{C}}\frac{1}{2}(1-\zeta-\sqrt{2\delta})\sum_{y_i\in\clY}\bbE_{\sfP_{R_i|\clE,D_i=1,G_i=y_i}}\sfP_{G_i|\clE,D_i=1}(y_i)\cdot\\ 
 & \quad \left[(1-\zeta-\sqrt{2\delta})\Vert\vph_{X_{\bar{C}}\tilde{X}_{\bar{C}}A|y_i,r_i,D_i=1,G_i=y_i}-\vph_{X_{\bar{C}}\tilde{X}_{\bar{C}}A|r_i,D_i=1,G_i=y_i}\Vert_1^2\right.\\
 & \quad \left.+ (\zeta/3-\sqrt{2\delta})\Vert\vph_{X_{\bar{C}}\tilde{X}_{\bar{C}}A|y^*,r_i,D_i=1,G_i=y_i}-\vph_{X_{\bar{C}}\tilde{X}_{\bar{C}}A|r_i,D_i=1,G_i=y_i}\Vert_1^2\right].
\end{align*}
where we have used \eqref{eq:Y=G=y} and Pinsker's inequality in the last line. Hence by triangle inequality we have,
\begin{align*}
\bbE_{i\in\bar{C}}\sum_{y_i\in \clY}\bbE_{\sfP_{R_i|\clE,1,y_i}}\sfP_{G_i|\clE,1}(y_i)\Vert\vph_{X_{\bar{C}}\tilde{X}_{\bar{C}}A|y_ir_i,1,y_i} - \vph_{X_{\bar{C}}\tilde{X}_{\bar{C}}A|y^*r_i,1,y_i}\Vert_1^2 \leq \frac{32\delta_1}{\zeta}.
\end{align*}
We note that $\vph_{X_{\bar{C}}\tilde{X}_{\bar{C}}Y_{\bar{C}}\tilde{Y}_{\bar{C}}ABZ_{\bar{C}}|y_ir_i,1,y_i}$ and $\vph_{X_{\bar{C}}\tilde{X}_{\bar{C}}Y_{\bar{C}}\tilde{Y}_{\bar{C}}ABZ_{\bar{C}}|y^*r_i,1,y_i}$ are pure states. Hence, using the Fuchs-van de Graaf inequality and Uhlmann's theorem, there exist unitaries $U_{y_ir_i}$ acting only on $Y_{\bar{C}}\tilde{Y}_{\bar{C}}BZ_{\bar{C}}$ such that
\begin{equation}\label{eq:y-y*-1}
\bbE_{i\in\bar{C}}\sum_{y_i\in \clY}\bbE_{\sfP_{R_i|\clE,1,y_i}}\sfP_{G_i|\clE,1}(y_i)\Vert\state{\vph}_{y_ir_i,1,y_i} - (\Id\otimes U_{y_ir_i})\state{\vph}_{y^*r_i,1,y_i}(\Id\otimes U^\dagger_{y_ir_i})\Vert_1 \leq \left(\frac{32\delta_1}{\zeta}\right)^{1/4}.
\end{equation}

Now consider the superoperator $\clO_{X_i}$ that measures the register $X_i$ and writes it in a different register.
\begin{align*}
\clO_{X_i}(\state{\vph}_{y_ir_i,1,y_i}) & = \sum_{x_i}\sfP_{X_i|\clE,y_ir_i,D_i=1,G_i=y_i}(x_i)\state{x_i}\otimes\state{\vph}_{x_iy_ir_i,1,y_i} \\
 & = \sum_{x_i}\sfP_{X_i|\clE,y_ir_i,D_i=1,G_i=y_i}(x_i)\state{x_i}\otimes\state{\vph}_{x_iy_ir_i} \\
\clO_{X_i}(\state{\vph}_{y^*r_i,1,y_i}) & = \sum_{x_i}\sfP_{X_i|\clE,y^*r_i,D_i=1,G_i=y_i}(x_i)\state{x_i}\otimes\state{\vph}_{x_iy^*r_i} 
\end{align*}
where we have made the observation that $\state{\vph}_{x_iy_ir_i,1,y_i}$ and $\state{\vph}_{x_iy^*r_i,1,y_i}$ are the same states as $\state{\vph}_{x_iy_ir_i}$ and $\state{\vph}_{x_iy^*r_i}$. By Fact \ref{hol} we can get,
\[ \bbE_{i\in\bar{C}}\Vert\sfP_{X_iG_iR_i|\clE,1} - \sfP_{G_iR_i|\clE,1}\sfP_{X_i|1,G_i}\Vert_1 \leq 2\sqrt{2\delta_1}.\]
Hence, for any value $Y_i=y_i$,
\begin{align*}
\bbE_{i\in\bar{C}}\Vert\sfP_{X_iG_iR_i|\clE,1} - \sfP_{G_iR_i|\clE,1}\sfP_{X_i|\clE,y_i,1,G_iR_i}\Vert_1 & \leq \bbE_{i\in\bar{C}}\left[\Vert\sfP_{X_iG_iR_i|\clE,1} - \sfP_{G_iR_i|\clE,1}\sfP_{X_i|y_i,1,G_i})\Vert_1\right. \\
& \quad + \left. \Vert\sfP_{G_iR_i|\clE,1}(\sfP_{X_i|y_i, 1, G_i} - \sfP_{X_i|\clE,y_i,1,G_iR_i})\Vert_1\right] \\
& \leq  \bbE_{i\in\bar{C}}\bigg[\Vert\sfP_{X_iG_iR_i|\clE,1} - \sfP_{G_iR_i|\clE,1}\sfP_{X_i|1,G_i}\Vert_1 \\
& \quad + \frac{2}{\zeta/3-\sqrt{2\delta}}\Vert\sfP_{X_iG_iR_i|\clE,1} - \sfP_{G_iR_i|\clE,1}\sfP_{X_i|1,G_i}\Vert_1\bigg] \\
& \leq \frac{8\sqrt{2\delta_1}}{\zeta}
\end{align*}
where we have used the fact that for any value $G_i=g_i$, we must have $\sfP_{Y_i|1,g_i}(y_i) \geq \zeta/3 - \sqrt{2\delta}$. Finally,
\[ \bbE_{i\in\bar{C}}\Vert\sfP_{X_iG_iR_i|\clE,1} - \sfP_{X_iY_iR_i|\clE,1}\Vert_1 \leq 2\sfP_{Y_iG_i|\clE,1}(Y_i \neq G_i) \leq \zeta/3 + \sqrt{2\delta}.\]
Observing that $\sfP_{X_iY_iR_i|\clE,1}$ is the same as $\sfP_{X_iY_iR_i|\clE}$ we get,
\[ \bbE_{i\in\bar{C}}\Vert\sfP_{X_iY_iR_i|\clE} - \sfP_{G_iR_i|\clE,1}\sfP_{X_i|\clE,y_i,1,G_iR_i}\Vert_1 \leq \frac{8\sqrt{2\delta_1}}{\zeta} + \frac{\zeta}{3} + \sqrt{2\delta}.\]
Using this and \eqref{eq:y-y*-1} we get,
\begin{align}
& \bbE_{i\in\bar{C}}\bbE_{\sfP_{X_iY_iR_i|\clE}}\Vert\state{\vph}_{x_iy_ir_i} - (\Id\otimes U_{y_ir_i})\state{\vph}_{x_iy^*r_i}(\Id\otimes U^\dagger_{y_ir_i})\Vert_1 \nonumber \\
& \leq \bbE_{i\in\bar{C}}\Bigg[\Vert\sfP_{X_iY_iR_i|\clE} - \sfP_{G_iR_i|\clE,1}\sfP_{X_i|\clE,y_i,1,G_iR_i}\Vert_1 + \Vert\sfP_{X_iY_iR_i|\clE} - \sfP_{G_iR_i|\clE,1}\sfP_{X_i|\clE,y^*,1,G_iR_i}\Vert_1 \nonumber \\
& \quad + \bbE_{\sfP_{G_iR_i|\clE,1}}\left\Vert\bbE_{\sfP_{X_i|\clE,y_ir_i,1,y_i}}\state{x_i}\otimes\state{\vph}_{x_iy_ir_i} - \bbE_{\sfP_{X_i|\clE,y^*r_i,1,y_i}}\Id\otimes U_{y_ir_i}\state{\vph}_{x_iy^*r_i}\Id\otimes U^\dagger_{y_ir_i}\right\Vert_1\Bigg] \nonumber \\
& =  \frac{16\sqrt{2\delta_1}}{\zeta} + \frac{2\zeta}{3} + 2\sqrt{2\delta} + \left(\frac{32\delta_1}{\zeta}\right)^{1/4} < \frac{7\zeta}{10} \label{eq:Y-U-dist}
\end{align}
where we have bounded the last term in the first inequality by applying Fact \ref{chan-l1} on \eqref{eq:y-y*-1} with $\clO_{X_i}$. Notice that we have also removed the conditioning $G_i\neq y^*$, since for $G_i=y^*$, the corresponding states are both $\ket{\vph}_{x_iy^*r_i}$.

From \eqref{eq:X-proj-dist} and \eqref{eq:Y-U-dist} we get,
\begin{align}
& \bbE_{i\in\bar{C}}\bbE_{\sfP_{X_iY_iR_i|\clE}}\left\Vert\frac{1}{\alpha_{r_i}}(\Pi_{x_ir_i}\otimes U_{y_ir_i})\state{\vph'}_{y^*r_i}(\Pi_{x_ir_i}\otimes U^\dagger_{y_ir_i}) - \state{\vph}_{x_iy_ir_i}\right\Vert_1 \nonumber \\
&\leq \bbE_{i\in\bar{C}}\bbE_{\sfP_{X_iY_iR_i|\clE}}\bigg[\left\Vert\frac{1}{\alpha_{r_i}}(\Pi_{x_ir_i}\otimes U_{y_ir_i})\state{\vph'}_{y^*r_i}(\Pi_{x_ir_i}\otimes U^\dagger_{y_ir_i}) - (\Id\otimes U_{y_ir_i})\state{\vph}_{x_iy^*r_i}(\Id\otimes U^\dagger_{y_ir_i})\right\Vert_1 \nonumber \\
& \quad  + \left\Vert(\Id\otimes U_{y_ir_i})\state{\vph}_{x_iy^*r_i}(\Id\otimes U^\dagger_{y_ir_i}) - \state{\vph}_{x_iy_ir_i}\right\Vert_1\bigg] \nonumber \\
& = \bbE_{i\in\bar{C}}\bbE_{\sfP_{X_iY_iR_i|\clE}}\bigg[\left\Vert\frac{1}{\alpha_{r_i}}(\Pi_{x_ir_i}\otimes \Id)\state{\vph'}_{y^*r_i}(\Pi_{x_ir_i}\otimes \Id) - \state{\vph}_{x_iy^*r_i}\right\Vert_1 \nonumber \\
& \quad  + \left\Vert(\Id\otimes U_{y_ir_i})\state{\vph}_{x_iy^*r_i}(\Id\otimes U^\dagger_{y_ir_i}) - \state{\vph}_{x_iy_ir_i}\right\Vert_1\bigg] \nonumber \\
& \leq \frac{7\zeta}{2} + \frac{7\zeta}{10} = \frac{21\zeta}{5}. \label{eq:XY-dist}
\end{align}

Using Markov's inequality on \eqref{eq:XY(R|X)}, \eqref{eq:XY(R|Y)}, \eqref{eq:X-proj-succ} and \eqref{eq:XY-dist}, we get an index $i\in\bar{C}$ such that the conditions \eqref{eq:XY(R|X)i}-\eqref{eq:XY-dist-i} for Lemma \ref{i-conds} hold.
\end{proof}

%--------------------------------------------

\section{Proof of parallel repetition theorem}
In this section we prove Theorem \ref{thm:par-rep}, whose statement is recalled below.
\parrep*

The proof of this theorem is very similar to that of the direct product theorem, so we shall only highlight points of difference. Whereas in the communication case, we started with an arbitrary distribution $p$ and defined distribution $q$ anchored on one side close to $p$, here we start with an already anchored distribution. To preserve similarity with the direct product proof, we shall consider $q$ to be anchored on the $\clY$ side here as well, but the proof goes through analogously for a distribution anchored on the $\clX$ side. We define the correlation-breaking variables and the joint distribution $\sfP_{XYDG}$ exactly as before.\footnote{The definition of $\sfP_{X_iY_iD_iG_i}$ in the previous section makes references to $p(x,y)$. Since there is no $p$ in the present case, $p(x,y)$ can simply be replaced by $q(x,y|y\neq y^*)$.}

We consider an entangled strategy $\clS$ for $G^k$, where Alice and Bob, with input registers $X=X_1\ldots X_k$ and $Y=Y_1\ldots Y_k$, initially share an entangled state, and perform unitaries $V^\text{Alice}$ and $V^\text{Bob}$ respectively on their parts of the entangled state and and their input registers. As before, conditioned on any value $DG=dg$, we define the following pure state representing $\clS$ after these unitaries:
\[ \ket{\theta}_{X\tilde{X}Y\tilde{Y}ABE'_AE'_B|dg} = \sum_{xy}\sqrt{\sfP_{XY|dg}(xy)}\ket{xxyy}_{X\tilde{X}Y\tilde{Y}}\otimes\ket{\theta}_{ABE_AE_B|xy}\]
where $AB$ are the answer registers which are measured in the computational basis by Alice and Bob to obtain their answers $(a,b)$, and $E'_AE'_B$ are some additional registers which are discarded. We shall use $\sfP_{XYAB|dg}$ to denote the distribution of $XYAB$ in $\ket{\theta}_{dg}$; $\sfP_{XYDGAB}$ is obtained by averaging over $dg$.  

Let the winning probability of of $\omega^*(G)$ be $1-5\eps$ for an appropriate $\eps$. We shall prove the following lemma, which is analogous to the direct product case. It is clear that the lemma implies
\[ \omega^*(G^k) \leq \left(1-\eps\right)^{\frac{\zeta^2\eps^4k}{\log(|\clA|\cdot|\clB|)}} = \left(1 - (1-\omega^*(G))^5\right)^{\Omega\left(\frac{\zeta^2 k}{\log(|\clA|\cdot|\clB|)}\right)}.\]
\vspace{-0.3cm}
\begin{lemma}\label{parrep-ind}
Let $\delta = \frac{\zeta^2\eps^4}{1440000}$ and $\delta' = \frac{\zeta^2\eps^4}{1440000\log(|\clA|\cdot|\clB|)}$. For $i\in[k]$, let $T_i$ denote the random variable $\sfV(X_i,Y_i,A_i,B_i)$, where $X_iY_iA_iB_i$ are according to $\sfP_{XYAB}$. Then there exist $\lfloor\delta'k\rfloor$ coordinates $\{i_1, \ldots, i_{\lfloor\delta' k\rfloor}\} \subseteq [k]$, such that for all $1 \leq r \leq \lfloor\delta'k\rfloor-1$, at least one of the conditions holds
\begin{enumerate}[(i)]
\item $\Pr\left[\prod_{j =1}^rT_{i_j} = 1\right] \leq (1-\eps)^{\delta k}$
\item $\Pr\left[T_{i_{r+1}} = 1 \middle| \prod_{j=1}^rT_{i_j} = 1\right] \leq 1-\eps$.
\end{enumerate}
\end{lemma}

As before, we shall consider that we have identified a set of coordinates $C = \{i_1,\ldots, i_t\}$ such that for all $1\leq r \leq t-1$, $\Pr\left[T_{i_{r+1}} = 1| \prod_{j=1}^r T_{i_j} = 1\right] \leq 1-\eps$ and $\Pr[\clE] = \Pr\left[\prod_{j=1}^tT_{i_j}=1\right] \geq (1-\eps)^{\delta k}$, and identify a $(t+1)$-th coordinate $i$. Let $E_A$ and $E_B$ to denote $A_{\bar{C}}E'_A$ and $B_{\bar{C}}E'_B$ respectively. We define the following state, which is $\ket{\theta}_{dg}$ conditioned on success in $C$:
\begin{align*}
& \ket{\vph}_{X\tilde{X}Y\tilde{Y}A_CB_CBE_AE_B|dg} \\
& = \frac{1}{\sqrt{\gamma_{dg}}}\sum_{xy}\sqrt{\sfP_{XY|dg}(xy)}\ket{xxyy}_{X\tilde{X}Y\tilde{Y}}\otimes\sum_{a_Cb_C:\sfV^t(x_C,y_C,a_C,b_C)=1}\ket{a_Cb_C}_{A_CB_C}\ket{\tilde{\vph}}_{E_AE_B|xya_Cb_C}.
\end{align*}
Here $\ket{\tilde{\vph}}_{E_AE_B|xya_Cb_C}$ is a subnormalized state satisfying $\Vert\ket{\tilde{\vph}}_{E_AE_B|xya_Cb_C}\Vert_2^2 = \sfP_{A_CB_C|xy}(a_Cb_C)$.

The following lemma is the analog of Lemma \ref{i-conds}, which we shall use to prove Lemma \ref{parrep-ind}.
\begin{lemma}\label{parrep-i-conds}
If $\Pr[\clE] \geq (1-\eps)^{\delta k}$, then there exist a coordinate $i\in\bar{C}$, a random variable $R_i =X_CY_CA_CB_CD_{-i}G_{-i}$, such that the following conditions hold:
\begin{enumerate}[(i)]
\item $\Vert\sfP_{X_iY_iR_i|\clE} - \sfP_{X_iY_i}\sfP_{R_i|\clE,X_i}\Vert_1 \leq \frac{7\eps}{150}$ 
\item $\Vert\sfP_{X_iY_iR_i|\clE} - \sfP_{X_iY_i}\sfP_{R_i|\clE,Y_i}\Vert_1 \leq \frac{7\eps}{150}$
\item There exist unitaries $\{U_{x_ir_i}\}_{x_ir_i}$ and $\{U_{y_ir_i}\}_{y_ir_i}$ respectively acting only on $X_{\bar{C}}\tilde{X}_{\bar{C}}E_A$ and $Y_{\bar{C}}\tilde{Y}_{\bar{C}}E_B$, such that
\[ \bbE_{\sfP_{X_iY_iR_i|\clE}}\left\Vert(U_{x_ir_i}\otimes U_{y_ir_i})\state{\vph}_{y^*r_i}(U^\dagger_{x_ir_i}\otimes U^\dagger_{y_ir_i}) - \state{\vph}_{x_iy_ir_i}\right\Vert_1 \leq \frac{36\eps}{5}.\]
\end{enumerate}
\end{lemma}

It is easy to see how this lemma implies Lemma \ref{parrep-ind}. As in the direct product case, Alice and Bob share $\ket{\vph}_{y^*r_i}$ as entanglement -- though in this case only one copy, as well as classical randomness with which they can produce $R^\text{Alice}_iR^\text{Bob}_i$ satisfying
\[ \Vert\sfP_{X_iY_iR^\text{Alice}_iR^\text{Bob}_i} - \sfP_{X_iY_iR_iR_i|\clE}\Vert_1 \leq \frac{7\eps}{30}.\]
Alice and Bob apply $U_{x_ir^\text{A}_i}$ and $U_{y_ir^\text{B}_i}$ according to their inputs and $R^\text{Alice}_i$ and $R^\text{Bob}_i$ respectively, on their registers registers $E_A$ and $E_B$ of $\ket{\vph}_{y^*r_i}$. They then measure in the computational basis on the $A_iB_i$ registers of resulting state, to give their outcomes $(a_i,b_i)$. $\Pr[T_i=1|\clE] \geq 1-\eps$ implies that the resulting strategy for $G$ has success probability $>(1 - 5\eps)$, a contradiction which lets us identify $i$ as the $(t+1)$-th coordinate.

The rest of the proof will be dedicated to showing Lemma \ref{parrep-i-conds}.
\begin{proof}[Proof of Lemma \ref{parrep-i-conds}]
We can prove
\begin{align}
\bbE_{i\in\bar{C}}\Vert\sfP_{X_iY_iR_i|\clE} - \sfP_{X_iY_i}\sfP_{R_i|\clE, X_i}\Vert_1 & \leq \frac{7\eps}{600} \\
\bbE_{i\in\bar{C}}\Vert\sfP_{X_iY_iR_i|\clE} - \sfP_{X_iY_i}\sfP_{R_i|\clE, Y_i}\Vert_1 & \leq \frac{7\eps}{600} \\
\bbE_{i\in\bar{C}}\bbE_{\sfP_{X_iY_iR_i|\clE}}\Vert\state{\vph}_{x_iy_ir_i} - (\Id\otimes U_{y_ir_i})\state{\vph}_{x_iy^*r_i}(\Id\otimes U^\dagger_{y_ir_i})\Vert_1 & \leq \frac{4\eps}{5} \label{eq:Y-U-dist-2}
\end{align}
exactly the same way as in the direct product case, except conditioning on $z_C$ is replaced by conditioning on $a_Cb_C$, which leads to the factor of $\log(|\clA|\cdot|\clB|)$. The rest of the proof will hence be spent getting Alice's unitaries $U_{x_ir_i}$.

Letting $\delta_1 = \delta + \delta'\log(|\clA|\cdot|\clB|)$, the following is derived analogously to the direct product case, except for the extra factor in the mutual information bound due to communication:
\begin{align}
\bbE_{i\in\bar{C}}\bbE_{R_i|\clE,D_i=1,G_i=y^*}\sfI(X_i:Y_{\bar{C}}\tilde{Y}_{\bar{C}}E_B)_{\vph_{r_i,D_i=1,G_i=y^*}} & \leq \frac{10\delta_1}{\zeta} \label{eq:X-info-2} \\
\bbE_{i\in\bar{C}}\Vert\sfP_{X_iR_i|\clE,y^*} - \sfP_{X_iR_i|\clE,1,y^*}\Vert_1 & \leq \frac{7\sqrt{2\delta_1}}{\zeta} \label{eq:XR-dist-1} \\
\bbE_{i\in\bar{C}}\Vert \sfP_{X_iR_i|\clE} - \sfP_{X_iR_i|\clE,1,y^*} \Vert_1 & \leq \frac{40\sqrt{2\delta_1}}{\zeta}. \label{eq:XR-dist-2}
\end{align}
From \eqref{eq:X-info-2}, by applying Pinsker's inequality, we get,
\[ \bbE_{i\in\bar{C}}\bbE_{\sfP_{X_iR_i|\clE,1,y^*}}\Vert \vph_{Y_{\bar{C}}\tilde{Y}_{\bar{C}}E_B|x_ir_i,1,y^*} - \vph_{Y_{\bar{C}}\tilde{Y}_{\bar{C}}E_B|r_i,1,y^*}\Vert_1 \leq \left(\frac{10\delta_1}{\zeta}\right)^{1/2} \]
Note that $\vph_{Y_{\bar{C}}\tilde{Y}_{\bar{C}}E_B|x_ir_i,1,y^*}$ is the same state as $\vph_{Y_{\bar{C}}\tilde{Y}_{\bar{C}}E_B|x_iy^*r_i}$. But $\vph_{Y_{\bar{C}}\tilde{Y}_{\bar{C}}E_B|r_i,1,y^*}$ is not the same state as $\vph_{Y_{\bar{C}}\tilde{Y}_{\bar{C}}E_B|y^*r_i}$, due to the averaging over $X_i$ being done with respect to $\sfP_{X_i|\clE,r_i,1,y^*}$ in one, and with respect to $\sfP_{X_i|\clE,y^*r_i}$ in the other. However, due to \eqref{eq:XR-dist-1} we can say,
\begin{align*}
& \bbE_{i\in\bar{C}}\bbE_{\sfP_{X_iR_i|\clE,1,y^*}}\Vert \vph_{Y_{\bar{C}}\tilde{Y}_{\bar{C}}E_B|x_iy^*r_i} - \vph_{Y_{\bar{C}}\tilde{Y}_{\bar{C}}E_B|y^*r_i}\Vert_1 \\
& \leq \left(\frac{10\delta_1}{\zeta}\right)^{1/2} + \bbE_{i\in\bar{C}}\Vert\sfP_{X_iR_i|\clE,1,y^*} - \sfP_{R_i|\clE,1,y^*}\sfP_{X_i|\clE,R_i,y^*}\Vert_1 \\
& \leq \left(\frac{10\delta_1}{\zeta}\right)^{1/2} + \bbE_{i\in\bar{C}}\Vert \sfP_{X_iR_i|\clE,y^*} - \sfP_{X_iR_i|\clE,1,y^*} \Vert_1 \\
& \leq \frac{2\sqrt{108\delta_1}}{\zeta}.
\end{align*}
Since $\ket{\vph}_{X_{\bar{C}}\tilde{X}_{\bar{C}}Y_{\bar{C}}\tilde{Y}_{\bar{C}}E_AE_B|y^*r_i}$ is a purification of $\vph_{Y_{\bar{C}}\tilde{Y}_{\bar{C}}E_B|y^*r_i}$ and $\ket{\vph}_{X_{\bar{C}}\tilde{X}_{\bar{C}}Y_{\bar{C}}\tilde{Y}_{\bar{C}}E_AE_B|x_iy^*r_i}$ is a purification of $\vph_{Y_{\bar{C}}\tilde{Y}_{\bar{C}}E_B|x_iy^*r_i}$, by the Fuchs-van de Graaf inequality and Uhlmann's theorem we can say that there exist unitaries $U_{x_ir_i}$ on $X_{\bar{C}}\tilde{X}_{\bar{C}}E_A$ such that
\[ \bbE_{i\in\bar{C}}\bbE_{\sfP_{X_iR_i|\clE,1,y^*}}\Vert\state{\vph}_{x_iy^*r_i} - (U_{x_ir_i}\otimes\Id)\state{\vph}_{y^*r_i}(U^\dagger_{x_ir_i}\otimes\Id)\Vert_1 \leq \left(\frac{2\sqrt{108\delta_1}}{\zeta}\right)^{1/2}\]
and by \eqref{eq:XR-dist-2} again,
\begin{align}
\bbE_{i\in\bar{C}}\bbE_{\sfP_{X_iR_i|\clE}}\Vert\state{\vph}_{x_iy^*r_i} - (U_{x_ir_i}\otimes\Id)\state{\vph}_{y^*r_i}(U^\dagger_{x_ir_i}\otimes\Id)\Vert_1 & \leq \left(\frac{2\sqrt{108\delta_1}}{\zeta}\right)^{1/2} + \frac{40\sqrt{2\delta_1}}{\zeta} \nonumber \\
& \leq 2\left(\frac{10800\delta_1}{\zeta^2}\right)^{1/4} \nonumber \\
& \leq \eps. \label{eq:X-U-dist-2}
\end{align}
Combining \eqref{eq:X-U-dist-2} and \eqref{eq:Y-U-dist-2} we get,
\[ \bbE_{i\in\bar{C}}\bbE_{\sfP_{X_iY_iR_i|\clE}}\left\Vert(U_{x_ir_i}\otimes U_{y_ir_i})\state{\vph}_{y^*r_i}(U^\dagger_{x_ir_i}\otimes U^\dagger_{y_ir_i}) - \state{\vph}_{x_iy_ir_i}\right\Vert_1 \leq \frac{9\eps}{5}.\]
The result then follows by Markov's inequality.
\end{proof}

%--------------------------------------------

\section*{Acknowledgements}
This work is supported by the National Research Foundation, including under NRF RF Award No. NRF-NRFF2013-13, the Prime Minister's Office, Singapore and the Ministry of Education, Singapore, under the Research Centres of Excellence program and by Grant No. MOE2012-T3-1-009 and in part by the NRF2017-NRF-ANR004 {\em VanQuTe} Grant.

\bibliographystyle{alpha}
\bibliography{bib-oneway}
\end{document}